\documentclass[11pt]{article}
\usepackage{amsfonts,amssymb,amsmath,amsthm,bbm}
\usepackage{graphicx, color, psfrag,hyperref,xcolor}

\usepackage[english]{babel}
\usepackage{color}

\newcommand{\ew}{\color{black}}

\newtheorem{theorem}{Theorem}
\newtheorem{definition}{Definition}
\newtheorem{remark}{Remark}

\newtheorem{proposition}{Proposition}
\newtheorem{lemma}{Lemma}

\def\d{\mathop{\textrm{\rm d}}\nolimits}                  
\def\exp{\mathop{\textrm{\rm exp}}\nolimits}              
\def\supp{\mathop{\textrm{\rm supp}}\nolimits}            
\def\Ext{\mathop{\textrm{\rm Ext}}\nolimits}            
\def\sf{\mathop{\textrm{\rm sf}}\nolimits}            

\newcommand\blue[1]{\textcolor{blue}{}}

\newcommand{\beq}{\begin{eqnarray}}
\newcommand{\eeq}{\end{eqnarray}}

\newcommand{\beqt}{\begin{eqnarray*}}
\newcommand{\eeqt}{\end{eqnarray*}}

\newcommand{\be}{\begin{equation}}
\newcommand{\ee}{\end{equation}}

\newcommand{\bl}{\begin{lemma}}
\newcommand{\el}{\end{lemma}}

\newcommand{\bt}{\begin{theorem}}
\newcommand{\et}{\end{theorem}}

\newcommand{\bd}{\begin{definition}}
\newcommand{\ed}{\end{definition}}

\newcommand{\bp}{\begin{proposition}}
\newcommand{\ep}{\end{proposition}}

\newcommand{\bi}{\begin{itemize}}
\newcommand{\ei}{\end{itemize}}

\newcommand{\ben}{\begin{enumerate}}
\newcommand{\een}{\end{enumerate}}

\newcommand{\Z}{\mathbb Z}

\newcommand{\E}{\mathbb E}

\newcommand{\bbP}{\mathbb P}

\newcommand{\mee}{\ensuremath{\mathcal{M}}}

\newcommand{\eps}{\ensuremath{\varepsilon}}

\newcommand{\bbp}{\mathbb P}

\newcommand{\calz}{\ensuremath{\mathcal{Z}}}
\newcommand{\calk}{\ensuremath{\mathcal{K}}}

\newcommand{\lf}{\ensuremath{\left(}}
\newcommand{\ri}{\ensuremath{\right)}}

\newcommand{\calr}{\ensuremath{\mathcal{R}}}

\newcommand{\un}{\ensuremath{\underline}}

\def\free{\mathop{\textrm{\rm free}}\nolimits}           
\def\d{\mathop{\textrm{\rm d}}\nolimits}                  
\def\sign{\mathop{\textrm{\rm sign}}\nolimits}                  
\def\sf{\mathop{\textrm{\rm sf}}\nolimits}            
\def\ex{\mathop{\textrm{\rm ex}}\nolimits}            
\def\emp{\mathop{\textrm{\rm emp}}\nolimits}           

\def\Ext{\mathop{\textrm{\rm Ext}}\nolimits}           
\def\ex{\mathop{\textrm{\rm ex}}\nolimits}            
\def\inte{\mathop{\textrm{\rm in}}\nolimits}            

\begin{document}

\title{{\bf On  Long Range Ising Models with Random Boundary Conditions}}
 
 \author{Eric O. Endo 
 \ew 
 \footnote{NYU-ECNU Institute of Mathematical Sciences at NYU Shanghai, 3663 Zhongshan Road North, Shanghai, 200062, China,
 \newline
 email: ericossamiendo@gmail.com}\\
 Aernout C.D.  van Enter \footnote{ Bernoulli Institute, University of Groningen, Nijenborgh 9, 9747AG,Groningen, Netherlands,
 \newline
 email: aenter@phys.rug.nl}\\ 
  Arnaud Le Ny \footnote{LAMA UMR CNRS 8050, UPEC, Universit\'e Paris-Est,  94010 Cr\'eteil, France,
 \newline
 email:  arnaud.le-ny@u-pec.fr}
 }

\maketitle

\begin{center}
{\bf Abstract} 
\end{center}
We consider  polynomial long-range Ising models in one dimension, with ferromagnetic pair interactions decaying with power $2-\alpha$ (for $0 \leq \alpha < 1$), and prepared with randomly chosen boundary conditions. We show that at low temperatures in the thermodynamic limit the finite-volume Gibbs measures do not converge, but have a distributional limit, the so-called metastate. We find that there is a distinction between the values of $\alpha$ less than  or larger than $\frac{1}{2}$. For moderate, or intermediate, decay $\alpha < \frac{1}{2}$, the metastate is very dispersed and supported on the set of all Gibbs measures, both extremal and non-extremal, whereas for slow decays $\alpha > \frac{1}{2}$ the metastate is still dispersed, but has its support just on the set of the two extremal Gibbs measures, the plus measure  and the minus measure. 

The former, moderate decays case, appears to be new and is due to the occurrence  of almost sure  boundedness of the random variable which is the sum of all interaction (free) energies between random and ordered half-lines,  when the decay is fast enough, but still slow enough to get a phase transition ($\alpha>0$); while the latter, slow decays  case, is more reminiscent of and similar to the behaviour of higher-dimensional nearest-neighbour Ising models with diverging boundary (free) energies. 

We leave the threshold case $\alpha=\frac{1}{2}$ for further studies.

\footnotesize

\vspace{2cm}

 {\em  AMS 2000 subject classification}: Primary- 60K35 ; secondary- 82B20

{\em Keywords and phrases}: Long-range Ising models, metastates, random boundary conditions.

\tableofcontents

\normalsize
\section{Introduction}

 We consider  polynomial long-range Ising models in one dimension, also called  Dyson or Dyson-Ising models in previous works\footnote{The term ``Dyson model" we have used before has turned out to be slighly controversial. Apart from the fact that the model was not really invented by Dyson -- as quite frequently  occurs in  model-naming --, another problem is that the terminology is more often used for the hierarchical version, which Dyson invented as an auxiliary model to study the polynomial long-range one. However, the alternative, ``long-range Ising model" which was used in other works, does not remove the ambiguity either, as this term is also often used either for mean-field, or for Kac-type interactions. We therefore chose here  either to add the polynomial characterization, or to sometimes keep only the  ``Long-range Ising model" terminology of Cassandro {\em et al.}, but we also will still  sometimes keep referring to "Dyson models" for the sake of brevity, although we are  aware of that  this   also is a bit problematic.}. They consist of one-dimensional ferromagnetic  Ising models with long-range  pair interactions, polynomially decaying with a decay power $2-\alpha$,  which recently have been studied for a variety of  reasons. They are of interest  especially when $0\le  \alpha <1$, as in this regime they  mimic the behaviour of high-dimensional models, but in a version in which one can continuously vary the dimension, by varying $\alpha$ between $0$ and $1$. They display a phase transition in this regime; however, in contrast to high-dimensional models\footnote{That is, for dimensions $d>2$,  as in $d=2$ it follows that there are no interface states  from the famous Aizenman-Higuchi Theorem.} there are no interface states, but  there exist maximally two extremal Gibbs measures, $\mu^+$ and $\mu^{-}$, which differ at low temperatures. Moreover, when the decay power equals $2$ (so $\alpha=0$), there exists a hybrid (random first-order) transition in temperature,  as well as an intermediate phase  with temperature-dependent slow decay of correlations in a temperature interval below this transition temperature.
 
 Recently, these polynomial long-range Ising models were shown to provide the first example of a Gibbs measure which is not a $g$-measure \cite{BEEL}. For a selection of further results   and variations thereof, around these one-dimensional long-range models, see \cite{ACCN,BEEKR,BS,CFMP,CMP17,CMPR,COP,Dys,EKRS,F74,FILS,FrSp,IN,
 Joh,Kac,Ken,LP}.

Spin systems with random boundary conditions are among the simplest examples of quenched disordered systems. In the case of nearest-neighbour ferromagnetic  Ising models on trees, they are equivalent to Edwards-Anderson spin glass models \cite{CCST}. For an analysis of what happens  for Ising models on  the regular lattice $\Z^{d}$, $d\geq 2$, a question which was raised in \cite{NS2}, see \cite{EMN,ENS1, ENS2}.

In this latter case it was proven that at low temperatures a sequence of finite-volume Gibbs measures, for growing volumes with  random boundary conditions which are drawn in a symmetric i.i.d. manner,  will not converge. In other words, there is ``Chaotic Size Dependence (CSD)''  \cite{NS2}, see also \cite{Ent}. Typically,  the random boundary condition selects, with probability about $\frac{1}{2}$, either a plus measure or a minus measure, depending on the overall sign of the random boundary term; and there is no interface in the system. The resulting metastate \cite{NS1}, which describes the {\em distribution} of the different possible limit points, is not trivial but dispersed, as there are two different limit points, namely the plus state and the minus state, which both occur with probability $\frac{1}{2}$.

We remind the reader that metastates, objects introduced in 1995/1996  by Newman and Stein \cite{NS96c} to provide a framework to investigate the arising of CSD (which they formalized for spin-glasses in 1992 \cite {NS2}), in general are  (random) measures (thus measures on measures) on a product space of  disorder variables which we denote by $\eta$'s and the (Gibbs) measures on the spin variables $\sigma$, describing the weight distribution over the possible limit Gibbs measures, whether obtained as a kind of translation  averages, as in \cite{NS1, NS3}  or via conditioning, as in  \cite{AW}. Due to the randomness of the Gibbs measures in the support of the metastates, in general the metastates thus become measures on measures on measures...See also \cite{Bov,N, NS3}, or the more recent review \cite{NRS} where metastates and also the role of random boundary conditions are mentioned. 

The notions of Chaotic Size-Dependence and metastates have been developed to describe spin glasses, but as there is a lack of tractable non-mean-field spin-glass models, toy examples, such as ferromagnets with random boundary conditions, can provide instructive  illustrations of these notions. Their marginal on the disorder variables is the disorder measure, whereas conditioned on the disorder, they provide a measure on the Gibbs measures for that disorder realisation. Some studies regarding metastates in mean-field settings  are \cite{BG, FKR,IK,Kul}. For the theory of infinite-volume Gibbs measures, we refer to \cite{Bov,EFS,FV,HOG,Kos}.

Let us also mention that by  imposing a Mattis disorder, $J'(i,j)(\eta)=J(i,j)  \eta_i \eta_j $, $\eta_i= \pm 1$ i.i.d., 
and by using the random  gauge transformation of the pair  ($\sigma$,$\eta$) into $\sigma'_i =\eta_i \sigma_i$, statements which hold for almost all boundary conditions $\eta$ for a ferromagnet become statements for a fixed boundary condition $\sigma'$, but now valid for almost all realizations of the (Mattis) coupling disorder. For more on Mattis spin-glasses, see {\em e.g.} \cite{CCST, Kul98, Mattis}.
 
In this paper we investigate how polynomial long-range Ising models behave under random boundary conditions, and more in particular  what their metastate description is in that situation. This case shows substantial simplification as compared to the general case, as neither the set of  Gibbs measures, nor  the limit distribution on them, depend on the disorder variables. Thus the metastate is a product measure of the measure on the disorder variables (which is an independent product measure on the sequences of pluses and minuses, describing the random boundary conditions) and a measure (distribution) on the Gibbs measures for the Dyson models. By abuse of terminology we will also call this last distribution the metastate.    

It will turn out that there is a difference between two cases, namely the cases where $\alpha $ is smaller or larger than $\frac{1}{2}$.  The difference between those two cases rests on the fact that boundary terms are almost surely finite in the first case,  but they diverge in the second one.  In both  cases the metastate is dispersed, that is, it is not concentrated on a single (pure or mixed) Gibbs measure. But in the case $\alpha <   \frac{1}{2}$, the metastate will be concentrated on mixed states, whereas, in the case $\alpha > \frac{1}{2}$, similarly to what happens in higher-dimensional short-range models \cite{EMN,ENS1, ENS2}, we have a non-trivial metastate concentrated on the two pure Gibbs states, a situation which in some sense might be expected to be more ``typical", as seemed suggested by the recent results in \cite{CJK}. 
The case $\alpha = \frac{1}{2}$ seems more similar to the larger-$\alpha$ (equals slower decay) regime, but we have not checked if all our estimates go through in that more sensitive case.  Our proof strategy in the slow-decay regime loosely follows \cite{ENS1}, with some modifications which are needed due to the lack of independence of boundary conditions between different volumes, and the difference in low-temperature excitations (contours) showing up in the expansions we used to show the stability of similar zero-temperature results, which we obtain in Section \ref{SectionToyModel} for a toy version of our model (Ground states with weak (low-T)  boundary conditions). 
We need on the one hand to show that the boundary (free) energy is not too large, so we can with large probability still apply a contour analysis. On the other hand, we will use our weak local limit theorem for the slow-decay case to derive that the (free) energy is not too small, so we can exclude mixed states appearing in the support of the metastates.

In the next sections we will first give a short overview of the general theory of  lattice models and metastates we will need, and then give some of the  probabilistic tools we use, as well as the main results we will obtain. After discussing a toy version of our problem in Section 4, in  Sections 5 and 6 
we describe the contour representation and related cluster expansion estimates.
Proofs are further given  in Section 7 and in an Appendix (Section \ref{Appendix}). Section 8 contains a short summary of our conclusions. 


\section{Ising Model, random boundary conditions  and metastates}

\subsection{(Polynomial) Long-range Ising models}


Given a finite volume $\Lambda$ in $\mathbb{Z}$, we consider  polynomial long-range Ising models, {\em i.e.} ferromagnetic long-range one-dimensional Ising models, with a Hamiltonian defined on $\Omega_{\Lambda}=\{-1,1\}^{\Lambda}$ with boundary condition  $\eta_{\Lambda^c} \in \Omega_{\Lambda^c}=\{-1,1\}^{\Lambda^c}$ given by

\begin{equation}\label{Ham}
H^{\eta}_\Lambda(\sigma_{\Lambda}) = - \frac{1}{2}\sum_{x,y \in \Lambda} J_{xy} \sigma_x \sigma_y - \sum_{\substack{\substack{x \in \Lambda\\ y \in \Lambda^{c}}}} J_{xy} \sigma_x \eta_y,
\end{equation}
where $\sigma_{\Lambda}=(\sigma_i)_{i\in \Lambda}\in \Omega_{\Lambda}$ and, for a fixed $J\ge 1$ and $0\le \alpha<1$, the interaction $(J_{xy})_{x,y\in \mathbb{Z}}$ is defined by
\begin{eqnarray*}\label{interaction}
J_{xy}=
\begin{cases}
J, &\text{ if }|x-y|=1,\\
\frac{1}{|x-y|^{2-\alpha}}, &\text{ if }|x-y|>1,\\
0, &\text{ if }x=y.
\end{cases}
\end{eqnarray*}

For a fixed inverse temperature $\beta>0$, the Gibbs specification is determined by a family of probability measures 
$\gamma=\{\mu^{(\cdot)}_\Lambda\}_{\Lambda\Subset \Z}$ defined by
\be \label{GibbsMea}
\mu^\eta_\Lambda (\sigma_{\Lambda})= 
\frac{1}{Z_\Lambda^\eta} e^{-\beta H^{\eta}_\Lambda (\sigma_{\Lambda})}
\ee
where 
$Z^{\eta}_{\Lambda}$ is the partition function defined as usual by
$
Z^{\eta}_{\Lambda} = 
\sum_{\sigma_{\Lambda}\in \Omega_{\Lambda}} e^{-\beta H^{\eta}_\Lambda (\sigma_{\Lambda})}$. We write it $Z^{f}$, and similarly $H_\Lambda^{f}$ for free b.c. ($\eta_i=0,\; \forall i \in \Lambda^c$).


Consider $\Omega=\{-1,1\}^{\mathbb{Z}}$, let $\mathcal{E}$ be the Borel sigma-algebra  on $\{-1,1\}$, and let $\mathcal{F}=\mathcal{E}^{\otimes \Z}$ be the product sigma-algebra on $\Omega$. We denote by $\mee_1(\Omega)$  the set of probability measures on the measurable space $(\Omega,\mathcal{F})$.

{A function $f:\Omega \to \mathbb{R}$ is said to be \emph{local} if there exists a finite set $D\subset \mathbb{Z}$ such that $\sigma_D=\sigma'_D$ implies $f(\sigma)=f(\sigma')$. Denote by $D_f$  the smallest set satisfying this property, called the \emph{dependence set} of the function $f$. On the set of local functions, we attach the supremum norm $\lVert f \rVert = \sup_{\sigma \in \Omega}|f(\sigma)|$ to each function. 

 For a fixed probability measure $\mu\in \mathcal{M}_1(\Omega)$, define the family of seminorms
$$
\lVert \mu \rVert_X = \sup_{\substack{\lVert f \rVert = 1 \\ D_f \subset X}}|\mu(f)|
$$
over all finite $X\subset \mathbb{Z}$. The weak topology is generated by the set of open balls
$$
B_X^{\varepsilon}(\mu)=\{\nu\in \mathcal{M}_1(\Omega): \lVert \nu-\mu \rVert_X<\varepsilon\},
$$
where $\varepsilon>0$ and $X$ is finite. A sequence $\mu_n\in \mathcal{M}_1(\Omega)$ weakly converges to $\mu$ if, and only if, $\lVert \mu_n-\mu \rVert_X \to 0$ for every finite set $X\subset \mathbb{Z}$. In this weak topology, the space $\mathcal{M}(\Omega)$ is compact.}

A Gibbs measure $\mu$ is defined to be a probability measure on $\mee_1(\Omega)$ whose conditional probabilities with boundary condition $\eta$ outside $\Lambda$, are of the form of the kernels $\mu^\eta_\Lambda$, and thus satisfy the DLR equations
\be \label{DLR}
\mu \mu^{(\cdot)}_\Lambda = \mu, \text{ for all } \Lambda \Subset \Z,
\ee
where $\Lambda \Subset \Z$ means that $\Lambda$ is a finite subset of $\Z$. We denote by $\mathcal{G}(\gamma)$ the (convex) set of Gibbs measures for a given specification (or interaction), and by ${\rm ex} \; \mathcal{G}(\gamma)$ the set of its extreme elements. 

It is known that there exist at low temperatures only two extremal  Gibbs measures for these long-range models in the phase transition region (slow polynomial decay, low temperature) \cite{BS,F74,HOG},  {\em the plus measure}  $\mu^+$ and  {\em the minus measure} $\mu^{-}$. These measures are obtained as weak limits with the all$+$ (resp. all$-$) boundary conditions and are trivially translation-invariant, so there is no non-translation-invariant Gibbs measure. The fact that they are the only extremal elements in the set of Gibbs measures implies in particular that any other Gibbs measure $\mu$ is a convex combination of them:
$$
\mu =\lambda  \mu^+ + (1- \lambda) \mu^-,
$$
where $\lambda\in [0,1]$.
The weights $\lambda$ represent the relative probabilities of typical sets for both these phases. In particular, all the Gibbs measures are translation invariant -- just as is the case in the $2d$ nearest neighbor (n.n.) Ising model -- and the ``Dobrushin boundary condition'' $\eta=\pm$ defined by $\eta_x=\sign(x)$ (and $\omega_0=1$) leads  to  the symmetric convex mixture 
$$
\mu^\pm := \lim_{N \to \infty} \mu_{\Lambda_N}^{\pm}=\frac{1}{2} \mu^+ + \frac{1}{2} \mu^-, \quad \text{ where } \Lambda_N=[-N,N] \cap \mathbb{Z}.
$$

There exist various proofs of the existence of a phase transition for these polynomial long-range Ising  models \cite{ACCN,CFMP,Dys, FILS, FrSp,Joh}. As we noted, it is also  known that there are no interface states.  
In situations where one imposes Dobrushin boundary conditions, the interface has either mesoscopic (when $0<\alpha <1$) or macroscopic (when $\alpha=0$) fluctuations \cite{CMPR}. In this paper we refine in some sense  the analysis and show that for long-range models in dimension one, even if there are no interface states, some inhomogeneities near boundaries could  still manifest themselves at other scales.


\subsection{Random boundary conditions and metastates}

For each $x\in \Z$, define $\bbP_x$ to be the Bernoulli distribution on $(\{-1,1\},\mathcal{E})$ given by
$$
\bbP_x(\eta_x=1)=\frac{1}{2},
$$
and let $\bbP=\otimes_{x\in \Z}\bbP_x$ be the product measure on $(\Omega,\mathcal{F})$.

For a fixed finite interval $\Lambda\Subset \Z$, the Hamiltonian (\ref{Ham}) with random boundary condition~$\eta$ is a random variable on $(\Omega,\mathcal{F},\bbP)$, taking values in the set of bounded functions on $\Omega_{\Lambda}$.

We are interested in their behaviour in the infinite-volume limit of a family of Gibbs measures $\gamma[\eta]=\{\mu^{\eta}_{\Lambda}\}_{\Lambda\Subset \Z}$. Although  the extremal Gibbs measures and their mixtures are not disorder-dependent (and moreover, they  are translation invariant), we will find a  metastate, living on them, which  still will be a random quantity, inheriting its randomness from its construction using  the random boundary conditions. As the Gibbs measures themselves will not converge, only a distributional limit, the metastate,  will be the limit object. This limit object, however, lives on a different probability space from the one we used to describe the random boundary conditions.

In general a measure-on-measures-valued map $\kappa$ which is an element of
$ \mee_1(\mee_1(\Omega))$ will be  called a \emph{metastate}, if the set of Gibbs measures corresponding to different $\eta$ obtain full mass, i.e.,
$$
\kappa[\eta](\mathcal{G}(\gamma[\eta]))=1.
$$

The theory of metastates has been developed for general (quenched) disordered systems, and in particular for spin-glasses. As in such models it was not a priori clear whether any kind of pointwise  thermodynamic limit for the finite-volume states could be defined, it was investigated whether a weaker sense of distributional convergence still might make sense. It has been shown, mainly by Newman and Stein, that a theory of such convergence can be set up. The distribution (measure) on the possible limit points is called the metastate. 

 Our example, built from  a very simple model (a one-dimensional ferromagnet with pair interactions), will provide a new type of metastate for lattice systems. Its simplicity allows us to avoid requiring the most general version of the theory. In particular, we can identify the possible limit points (the plus and  the minus measure), and  due to the ferromagnetic character of the model we can avoid most of the measurability questions which the general theory was designed to address. Moreover, the measures themselves do not depend on the disorder, avoiding thus other potential measurability issues.

In the general theory, one way to construct a metastate is due to Aizenman and Wehr \cite{AW}. For a fixed Gibbs measure with random boundary condition $\mu_{\Lambda}^{\eta}$, consider the product measure
$$
\mathcal{K}_{\Lambda} = \bbP \otimes \delta_{\mu_{\Lambda}^\eta}.
$$
By Theorem 6.2.8 in   Bovier's book \cite{Bov}, for some increasing and absorbing sequence $(\Lambda_n)_{n\ge 1}$, i.e., $\Lambda_n \subset \Lambda_{n+1}$ and $\Lambda_n\uparrow \Z$, the weak limit
$$
\lim_{n\to \infty}\mathcal{K}_{\Lambda_n} = \mathcal{K}
$$
exists,  and its regular conditional distribution $\kappa=\mathcal{K}(\cdot|\mathcal{F}\times \Omega)$ is a metastate. This $\kappa$ is often called the \emph{Aizenman-Wehr metastate.} We also introduce $\kappa_\rho$ to serve as an {\em a priori} measure on the metastate space, 
\begin{equation}\label{kappa-rho}
\kappa_\rho= \lim_n \mathbb{P} \otimes \delta_{\rho_{\Lambda_n}}.
\end{equation}

 We will denote the $\bbP$-average by $K = \bbP \kappa$, and will by abuse of terminology call both $K$ and $\kappa$ the metastate. As in our example the $\kappa$ is constant and independent of the disorder parameter $\eta$, so we have a product state, ($K = \bbP \otimes \kappa$), this should not lead to confusion.  

Another way to construct metastates (in quite  broad generality) is due to Newman and Stein \cite{NS96c, NS1, NS4}, using what they called the \emph{empirical metastate}. Define the empirical measures $\kappa^{{\rm emp}}_N[\eta]:=\kappa_N[\eta]$ given by
\be \label{empirical}
\kappa_N[\eta](B) = \frac{1}{N} \sum_{n=1}^N \delta_{\mu_{\Lambda_n}^{\eta}\in B}.
\ee
Newman and Stein \cite{NS4} proved that, for sufficiently sparse sequences $(\Lambda_{N_k})_{k\ge 1}$, the limit point
\be \label{empMet}
\lim_{k\uparrow \infty}\kappa_{N_k}[\eta] = \kappa
\ee
exists for $\bbP$-a.e. $\eta \in \Omega$. Newman and Stein conjectured that it is not necessary to use sparse subsequences to have convergence. However, K\"ulske \cite{Kul} disproved this conjecture when he exhibited an example where the convergence almost-surely {\em only}  holds for sparse subsequences.  Moreover, under this construction, one also recovers the {\em a priori} metastate as

\begin{equation}\label{kappa-rho2}
\kappa_\rho (\cdot):= \lim_N \frac{1}{N} \sum_{n=1}^N \delta_{\rho_{\Lambda_n} \in \cdot}.
\end{equation}

 For spin models, including  our  long-range Ising models, for which ${\rm ex} \ \mathcal{G}(\gamma)=\{\mu^+,\mu^-\}$, the extremal elements of the set of Gibbs measures always have the property that they are selected  by  ``typical'' boundary conditions:
$$
\lim_{\Lambda\uparrow \Z} \mu^\eta_\Lambda = \mu, \quad {\rm for} \; \mu \text{-a.e. } \eta.
$$
However, this is far from being the case when  $\mu$ is a non-trivial mixture: the boundary condition will typically select one of the (here two) extremal phases, with a probability corresponding to their relative weights in the mixture.

Moreover, the behaviour of these limits can be more complex  when the boundary condition is not typical for either of the extreme phases, for example in the case of random  {\em incoherent} boundary conditions drawn from i.i.d. sampling, which is the case we will consider here.

At a few occasions, when needed, we denote by $\kappa^{\emp}$ this empirical metastate (\ref{empMet}) of Newman and Stein. In \cite{ENS1}, concerning the standard $2d$ n.n. ferromagnetic Ising model, these empirical measures have been used to describe the almost-sure structure of the set of limiting points of the finite-volume Gibbs measures, possibly along sparse enough sequences of squares $(\Lambda_{N_k})_{N_k\in \mathbb{N}}$. 
{There it was proved that under 
the sparseness condition $N_k \geq k^{2+\varepsilon}$, at low enough temperatures the metastates concentrate on $\{ \mu^-,\mu^+\}$, and also the null-recurrent character of other Gibbs measures, different from the plus and minus phases, was shown without such a sparseness condition.}

A quick way to describe this asymptotic non-triviality (which implies   {\em Chaotic Size Dependence}) is to consider that  convergence to the infinite-volume mixed metastate, which is  concentrated on two pure states, holds:
\be\label{metaconcentrated}
\lim_{k\to \infty}\kappa_{N_k}[\eta] =\frac{1}{2} \delta_{\mu^{+}} + \frac{1}{2} \delta_{\mu^{-}}.
\ee
As described  in Remark 3.2 in \cite{ENS1}, (\ref{metaconcentrated}) does not exclude other measures as almost-sure limit points provided that other (non-sparse) sequences of squares are taken. Their conjecture is in fact that at low enough temperature, the set of all weak limit points coincides $\mathbb{P}$-a.s. (thus for $\mathbb {P}$-almost all $\eta$) with the set $\mathcal{G}(\gamma[\eta])$. It is also added there that in dimension 3, for sequences of (hyper-)cubes, the set of limit points  is expected to coincide with the set of all translation-invariant Gibbs measures,  while in dimension higher than 3, with the set $\{\mu^-,\mu^+\}$ only.

Our main estimates  concern  random (free) energy differences between systems with free and with  random boundary conditions. These  differences will enter in the expression for the weights of the convex mixture and will behave differently depending on the speed of decay of the long-range interaction, which will create the different metastate dispersions (full dispersion on the full set of Gibbs measures, or minimal dispersion on the set of the two extremal and ``maximal'' Gibbs measures $\{\mu^-,\mu^+\}$. Indeed, depending on the decay parameter $\alpha$, random boundary energies will be almost surely finite  or infinite (c.q. diverging), giving rise either to two equal weights in the metastate (and thus to (\ref{metaconcentrated})), or to  weights leading to a dispersion over all the mixtures, for sufficiently sparse  volume subsequences (see Theorem \ref{thm1}).


We remark that recently some general properties of metastates have been studied  by Cotar {\em et al.} \cite{CJK}. In that work it was shown how to obtain a metastate supported by pure states, starting from a metastate supported by mixed states.

 In our case we have two situations, one where the metastates on the pure states are  the relevant ones, and another one, where the metastates live on mixtures of pure Gibbs states. 
Our example, where the phenomenon of a  dispersion over the set of all Gibbs measures occurs,  appears to be new for lattice systems. A similar dispersion has been described in the case of bulk disorder for  Hopfield models with a finite number of patterns by K\"ulske in \cite{Kul}. That model being of Mean-Field type, it lacks  boundary terms and the disorder is introduced in the bulk. It therefore seems rather different from our random-boundary-condition long-range model.

Note also another recent paper \cite{Read}, in which
 a class of one-dimensional long-range (in that case spin-glass) examples was considered with the help of the metastate formalism.
 
 {Very recently Chatterjee \cite{Chat23} has studied the low-temperature $d$-dimensional ($d \geq 2$) nearest-neighbour Ising model in a random field which is scaled in such a way that the total energy of the external-field term remains bounded (has a bounded variance) in the infinite-volume limit. In this model he obtains a very similar behaviour as occurs in the Dyson models with 
 $ 0 <  \alpha < \frac{1}{2} $:  The weights of the plus and minus states in the mixture are random variables, which will have a non-trivial limit distribution.  His statement that "the quenched distribution of the overlaps is not self-averaging" in our terms corresponds to the statement that there is a dispersed metastate, living on mixed states. 
 His result that these weights are given in terms of a Gaussian distribution however  has no analogue in our Dyson models.  }

\section{Non-trivial metastates and new dispersed behaviour}

Our main observation is then that if we take a sequence of increasing intervals 
$[-N_k, N_k]$ with $k\ge 1$, and consider a sufficiently sparse sequence, the metastate in the case $0 \leq \alpha < \frac{1}{2}$ has its support on (infinitely many) mixed states, while in case $\frac{1}{2}<\alpha < 1$, the metastate is supported by  the two extremal states and equals $\frac{1}{2} (\delta_{\mu^+} + \delta_{\mu^{-}})$, similarly to what happens in the higher-dimensional short-range cases treated in  \cite{EMN,ENS1,ENS2}.

The reason is that for intermediate decays  $0 \leq \alpha < \frac{1}{2}$ the boundary-energy  terms, even between two infinite half-lines, thus ``for infinite length of the interval", are almost surely finite;  in fact they approach a random variable $W^\eta$ which still is almost surely finite (but  drawn from a distribution which has some support outside any finite interval), whereas in the second situation $\frac{1}{2}< \alpha < 1$ the 
boundary-(free-)energy terms
 diverge and the (free) energy difference between plus and minus becomes infinitely large, almost surely and will generate weights equally concentrated on the extremal states $\mu^\pm$. 
 
 In contrast to it,  finite-energy boundary terms will generate mixed measures, which are absolutely continuous with respect to the symmetric mixture $\frac{1}{2}(\mu^+ + \mu^{-}) $ obtainable with free boundary conditions, compare \cite{BLP}. 
The limit distribution of the free boundary energy $F(W^{\eta})$ then translates into a distribution on the weights on the mixtures, formally proportional to the constrained partition functions $Z^{\pm,\eta} = e^{-\beta F(W^{\pm, \eta})}$.
 Therefore we obtain a metastate supported by (infinitely many) mixtures of $\mu^+$ and $\mu^{-}$, while in the very slow decay case $\frac{1}{2}< \alpha < 1$, infinite (diverging) boundary conditions give rise to a metastate supported on the two pure states $\mu^+$ and $\mu^{-}$.

For $\alpha = \frac{1}{2}$ the expected boundary energy also diverges, but slowly (logarithmically).

\begin{remark}  In the  higher-dimensional nearest-neighbour case, the choice of sparse sequences was sometimes needed to exclude null-recurrence of the set of mixed states.
In our long-range Ising model situation it is the case that  without imposing a sparsity condition we will always have such null-recurrence. Nevertheless, we need the sparsity in the sequence of intervals for a number of different reasons, namely to be able to treat the boundary-condition terms  for different intervals  as approximately independent, as well as for avoiding the null-recurrence which is due to the step-by-step adding of small energy amounts by adding sites to the interval one by one. Moreover,  sparsity will also guarantee that summability conditions hold which we can use for Borel-Cantelli arguments. 
\end{remark}

The main result of our our paper is the following:

\begin{theorem}\label{thm1}
For a fixed $0\le \alpha<1$,  different from $\alpha = \frac{1}{2}$, at sufficiently low temperatures, the metastate $\kappa$ is given by the following expression:
\begin{enumerate}
\item[(a)] For $\frac{1}{2}<\alpha <1$,
$$
\kappa
 = \frac{1}{2} \delta_{\mu^{-}} + \frac{1}{2} \delta_{\mu^{+}}.
$$

\item[(b)] For  $0 \le \alpha < \frac{1}{2}$,
let $\mu_\lambda = \lambda \mu^- + (1-\lambda) \mu^+$ be a generic mixed Gibbs measure,  where $\lambda \in [0,1]$. Then there is a non-trivial distribution $\d\lambda$ with support on the whole  interval $[0,1]$ (inducing a non-trivial distribution on the mixed Gibbs measures) such that 
$$
\kappa
 = \int_0^1 \delta_{\mu_{\lambda}} \kappa_\rho(\d\lambda)
$$
where $\kappa_\rho$ is the trivial, a priori metastate provided by (\ref{kappa-rho}) and (\ref{kappa-rho2}).
\end{enumerate}
\end{theorem}

For $\varepsilon<1-\alpha$, define a sequence $(m_k)_{k\ge 1}$ of positive integers satisfying
$$
\sum_{\ell = m_k}^{\infty} e^{-\ell^\varepsilon}<\frac{1}{k^2}.
$$
The reason we will need  the inequality above is that we can use it  to  restrict the boundary conditions $\eta$ to be ``good'',  with large enough probability,  to bound the boundary energy. We can see the precise statement in Proposition \ref{prop:goodeta}.

Theorem \ref{thm1} item (a) follows directly from the following theorem.

\begin{theorem}\label{thma}
For every $1/2<\alpha<1$ and $a>0$ there exists a $\beta_1=\beta_1(a)$ such that for
every $\beta>\beta_1$ the set of all weak limit points of any sequence $\{\mu^{
\eta}_{\Lambda_{N_k}}\}_{k\ge 1}$, with 
$$
N_k>\max \left\{k^{(\alpha-\frac{1}{2})^{-1}+a}, m_k^{2\left(\alpha-\frac{1}{2}\right)^{-1}} \right\}
$$
is $\{\mu^+,\mu^-\}$, $\mathbb{P}$-a.s.
\end{theorem}

The next theorem shows the null-recurrent character of all measures different from both $\mu^+$ and $\mu^-$. We use the notation $\bar{B}$ and $B^0$ for the weak closure and the weak interior of $B$, respectively.

\begin{theorem}\label{thmnull}
For every $1/2<\alpha<1$,  there exists $\beta_2$ such that for any $\beta>\beta_2$ and any set $
B\subset \mathcal{M}_1(\Omega)$, one has
$$
\lim_{N\to \infty} \kappa_N[\eta](B)=
\begin{cases}
0 &\text{ if } \mu^+,\mu^- \notin \bar{B}\\
\frac{1}{2} &\text{ if }\mu{\pm}\in B^0 \text{ and }\mu^{\mp}\in \bar{B}\\
1 &\text{ if }\mu^+,\mu^-\in B^0
\end{cases}
$$
$\mathbb{P}$-almost surely.
\end{theorem}

Both Theorem \ref{thma} and \ref{thmnull} follow from the following proposition below. 


\begin{proposition}\label{mainprop}
For every $1/2<\alpha<1$ and $\varepsilon>0$, there exists a $\beta_0=\beta_0(\varepsilon)$ such that for every $\beta>\beta_0$, $\tau>0$, and $X\Subset \mathbb{Z}$,
\be\label{eq:mainprop}
\limsup_{N\to \infty}N^{\alpha-\frac{1}{2}-\varepsilon}\mathbb{P}\left( \lVert \mu^{\eta}_{\Lambda_N}-\mu^+ \rVert_X \land \lVert \mu^{\eta}_{\Lambda_N}-\mu^- \rVert_X \ge \tau\right)<\infty.
\ee
\end{proposition}

To prove Proposition \ref{mainprop}, we apply the following weak Local Limit Theorem.

\begin{proposition}\label{LLT}\cite{ENS1}
Let $(X_n)_{n\in \mathbb{N}}$ be a sequence of random variables and denote by $\psi_n(t)$ the corresponding characteristic functions $\psi_n(t)=\mathbb{E}(e^{itX_n})$. If $(A_n)_{n\in \mathbb{N}}$, $(\delta_n)_{n\in \mathbb{N}}$ and $(\tau_n)_{n\in \mathbb{N}}$ are strictly positive sequences of real numbers satisfying the assumptions
\begin{enumerate}
	\item[(i)] $\limsup_{n\to \infty}A_n\int_{-\tau_n}^{\tau_n}|\psi_n(t)|dt \le 2\pi$
	\item[(ii)] There exists $k>1$ such that $\lim_{n\to \infty}\frac{A_n}{\delta_n^k \tau_n^{k-1}}=0$
\end{enumerate}
then
\be \label{prescribed}
\limsup_{n\to \infty}\frac{A_n}{\delta_n}\mathbb{P}(a\delta_n \le X_n\le b\delta_n)\le b-a
\ee
for any $a<b$.
\end{proposition}

The proofs of Theorem \ref{thma} and \ref{thmnull} are in Appendix 1 and that of Theorem \ref{thm1} in Appendix 3. We aim to prove Proposition \ref{mainprop} in the rest of  the paper.

\section{Toy model: Ground states with weak boundary conditions
}\label{SectionToyModel}

Consider the interval $\Lambda=\Lambda_N=[-N,N] \cap \mathbb{Z}$. The main mathematical work consists in investigating the asymptotic behaviour of random variables (random boundary free energies) expressed in terms of 
partition functions of some constrained ensembles with $\eta$ random boundary condition also related to the random variables
$$
W_N^\eta :=W_N^\eta(+) = \sum_{x \in \Lambda, y \in \Lambda^c} J_{xy} \eta_y.
$$
called {\em random boundary energies}, as their asymptotic behaviours will determine the weights entering the metastate decomposition.

We start with considering a simpler  toy version of our problem, by considering {\em ground states with weak boundary conditions}; that is, we consider an interval in which all bond strengths inside the interval are infinite - so the temperature is zero -, interacting with finite - thus positive-temperature - random boundary terms. To do so, we need to be able to consider in (\ref{GibbsMea}) $\beta=\infty$ for the interaction of pairs $\{x,y\}$ inside the volume~$\Lambda$ and $\beta < \infty$ for the interaction terms such that one site is inside $\Lambda$ and the other outside of it (cross-boundary interactions).

To consider the zero-temperature case 
 in the bulk, for edges contained inside the volume $\Lambda$, we first consider the Hamiltonian  with free boundary condition  and introduce -- for this intermediate model only -- an auxiliary temperature $\beta_1$ in the corresponding finite-volume Gibbs measures $\mu_\Lambda^f$ in (\ref{GibbsMea}),
$$
\tilde{\mu}_{\Lambda}^{\eta}(\sigma) = \frac{1}{\tilde{Z}^{\eta}_{\Lambda}}e^{-\beta_1 H^{f}_{\Lambda}(\sigma) + \beta W^{\eta}_{\Lambda}(\sigma)}.
$$

The corresponding infinite-volume Gibbs measure is the {\bf {\em symmetric}} mixture of $\tilde{\mu}^+$ and $\tilde{\mu}^-$ and by taking $\beta_1=\infty$, we get the zero-temperature picture with the symmetric  mixture of the Dirac measures on\footnote{So that one gets a convex combination of Dirac measures on Dirac measures, seen as a measure (the convex combination) on measures (the first Dirac measures) on measures (the Dirac measures on the ground states).} the ground states $\delta_+$ and $\delta_-$ (by  e.g.  an easy adaptation of Theorem 6.9 in \cite{HOG}  to this one-dimensional long-range set-up). To get the intermediate  model, we re-introduce our finite temperature $\beta >0$ on the cross-boundary edges only, to define the Gibbs measures with random b.c. by (\ref{GibbsMea}), in such a way that we only need to consider  the plus and minus configuration in the interval $[-N,N]$ on~$\Z$, interacting with a random configuration $\eta$ by  interaction terms  $J_{xy}$  of finite strength. Afterwards, to get independent r.v.'s in the bulk inside $\Lambda_N$, we put $\beta_1=\infty$, inside this volume. 

Let us define the plus configuration $\sigma^+=(\sigma^+_x)_{x\in \mathbb{Z}}$ by $\sigma^+_x= 1$, and  the minus configuration $\sigma^-=(\sigma^-_x)_{x\in \mathbb{Z}}$ by $\sigma^-_x= -1$. \\
For simplification, define 
$$
\delta_{\pm}(\sigma)=\delta_{\sigma^{\pm}}(\sigma)=\mathbbm{1}_{\sigma=\sigma^{\pm}}.
$$

	For $\beta_1 = \infty$, the 
finite-volume Gibbs measure on $\Omega_{\Lambda}$ with random boundary condition  $\eta$ is, for $\mathbb{P}$-a.e.($\eta)$, asymptotically equivalent to the finite-volume Gibbs measure
$$
\bar{\mu}^{\eta}_{\Lambda_N}(\sigma)=\frac{e^{- \beta W_N^\eta}}{e^{-\beta W_N^{\eta}}+e^{\beta W_N^{\eta}}}\delta_{+}
+ \frac{e^{\beta W_N^\eta}}{e^{-\beta W_N^{\eta}}+e^{\beta W_N^{\eta}}}\delta_{-}.
$$


%
%

By studying the asymptotics of the cumulative random energies $W_N$'s we get then our metastate result for this toy-model (Ground states with weak boundary conditions) : 

\begin{proposition}\label{toy_model_prop}
\begin{enumerate}
\item[(a)] For every $\frac{1}{2}<\alpha <1$ and $\beta>\beta_2 >0$  , the metastate $\kappa$ is given by the following expression:
	$$
	\kappa = \frac{1}{2} \delta_{\delta_{+}} + \frac{1}{2} \delta_{\delta_-}.
	$$
\item[(b)] For every $0 \leq \alpha <\frac{1}{2}$ and  $\beta>\beta_3> 0$
 , there is a non-trivial distribution $\d\lambda$ with support on the whole  interval $[0,1]$ inducing a non-trivial metastate which has its support on  all mixtures $\bar{\mu}_{\lambda} = \lambda \delta_{+} + (1 - \lambda) \delta_{-}$
$$
\kappa
 = \int_0^1 \delta_{\bar{\mu}_{\lambda}} \kappa_\rho(\d\lambda)
$$
where $\kappa_\rho$ is the trivial, a priori, metastate provided by (\ref{kappa-rho}) and (\ref{kappa-rho2}).
\end{enumerate} 
\end{proposition}


The proof of Proposition \ref{toy_model_prop} is in Appendix 2. To prove item (a), the following proposition, that is Proposition \ref{mainprop} for the toy model, is required.

\begin{proposition}\label{toy_mainprop}
For every $1/2<\alpha<1$, $\beta>0$, $\tau>0$, and $X\Subset \Z$,
$$
\limsup_{N\to \infty} N^{\alpha -\frac{1}{2}}\mathbb{P}\left( \lVert \bar{\mu}^{\eta}_{\Lambda_N} - \delta_+ \rVert_X \land \lVert \bar{\mu}^{\eta}_{\Lambda_N} - \delta_- \rVert_X \ge \tau  \right)<
\infty.
$$
\end{proposition}

\begin{proof}
We have that for every local function $f$ on $\Omega$,
\be \label{GS}
\bar{\mu}^{\eta}_{\Lambda_N}(f)=\frac{f(\sigma^+)}{1+e^{-2  \beta W_N^{\eta}}}+ \frac{f(\sigma^-)}{1+e^{2  \beta W_N^{\eta}}}.
\ee
Thus, one has directly
$$
\lVert \bar{\mu}^{\eta}_{\Lambda_N} - \delta_+ \rVert_X = \frac{2}{1+e^{2  \beta W_N^{\eta}}} \quad \text{and} \quad
\lVert \bar{\mu}^{\eta}_{\Lambda_N} - \delta_- \rVert_X = \frac{2}{1+e^{-2  \beta W_N^{\eta}}}
$$
for every $X\Subset \mathbb{Z}$. Then,
\be\label{toy_w}
\mathbb{P}\left( \lVert \bar{\mu}^{\eta}_{\Lambda_N} - \delta_+ \rVert_X \land \lVert \bar{\mu}^{\eta}_{\Lambda_N} - \delta_- \rVert_X \ge \tau  \right) = \mathbb{P}\left( \left| W_N^{\eta} \right|\le \frac{1}{2 \beta}\ln\left( \frac{2-\tau}{\tau} \right) \right).
\ee

Denote by $\psi_N(t):=\mathbb{E}\left(e^{- i t W_N^\eta} \right)$ the characteristic function of the random variable $W^{\eta}_N$. Since the $\eta_j$ are independent, the characteristic function of $W_N^\eta$ becomes
\begin{eqnarray*}
\psi_N(t) &=& \prod_{y\in \Lambda^c_N}  \mathbb{E}\left( e^{-it\sum_{x\in \Lambda_N}J_{xy}\eta_j} \right)=\prod_{y\in \Lambda^c_N} \frac{1}{2}\left( e^{-it\sum_{x\in \Lambda_N}J_{xy}} + e^{it\sum_{x\in \Lambda_N}J_{xy}} \right) \\ & =&\prod_{y\in \Lambda^c_N} \cos\left( t\sum_{x\in \Lambda_N}J_{xy} \right).
\end{eqnarray*}

To get an upper bound, use that there exists $0<\delta<\pi/2$ such that, for every $|u|<\delta$, 
we have $\cos(u)\le e^{-u^2/2}$.  For all $\alpha \in (1/2,1)$, choose
\be \label{tau}
\tau_N=\frac{\delta(1-\alpha)}{2-\alpha}N^{1-\alpha}.
\ee

Since, for all $y>N$,
$$
\sum_{x\in \Lambda_N}\frac{1}{|x-y|^{2-\alpha}}\le \frac{2-\alpha}{1-\alpha}(j-N)^{\alpha-1},
$$
and, for all $y<-N$,
$$
\sum_{x\in \Lambda_N}\frac{1}{|x-y|^{2-\alpha}}\le \frac{2-\alpha}{1-\alpha}(-y-N)^{\alpha-1},
$$
we have, for all $|y|>2N$,  all $\alpha \in (1/2,1)$, 
$$
\sum_{x\in \Lambda_N}J_{xy}<\frac{\delta}{\tau_N},
$$
concluding that
$$
t\sum_{x\in \Lambda_N}J_{xy}<\delta \quad \text{for every } 0<t<\tau_N \text{ and } |y|>2N.
$$
Now, choosing
$$ 
A_N=\left( \sum_{|y|>2N}\left( \sum_{x\in \Lambda_N}J_{xy} \right)^2 \right)^{1/2},
$$
we conclude
\begin{align*}
\int^{\tau_N}_{-\tau_N}|\psi_N(t)|\d t &= 2\int^{\tau_N}_0 \prod_{y\in \Lambda^c_N} \left| \cos\left( t\sum_{x\in \Lambda_N}J_{xy} \right) \right|\d t\\
&\le 2 \int^{\tau_N}_0\prod_{|y|> 2N} \left| \cos\left( t\sum_{x\in \Lambda_N}J_{xy} \right) \right|\d t\\
&\le 2\int_0^{\tau_N}\prod_{|y|> 2N}\exp\left( -\frac{t^2}{2}\left(\sum_{x\in \Lambda_N}J_{xy}\right)^2 \right)\d t\\
&=\frac{2}{A_N}\int_{0}^{\tau_N A_N} e^{-t^2/2}\d t\\
&\le \frac{2\pi}{A_N}.
\end{align*}


Therefore
\be \label{Car}
\limsup_{N\to \infty}A_N\int_{-\tau_N}^{\tau_N}|\psi_N(t)|\d t \le {2\pi}.
\ee
Let us compute an upper bound for $A_N$.
\begin{align*}
 \sum_{|y|>2N}\left( \sum_{x\in \Lambda_N}J_{xy} \right)^2&=
 \sum_{y>2N}\left( \sum_{x\in \Lambda_N}J_{xy} \right)^2 +
\sum_{y<-2N}\left( \sum_{x\in \Lambda_N}J_{xy} \right)^2\\
&\le \sum_{y>2N}(2N+1)^2\frac{1}{(y-N)^{4-2\alpha}} + \sum_{y<-2N}(2N+1)^2\frac{1}{(-y-N)^{4-2\alpha}}\\
&= 2(2N+1)^2\sum_{y>2N}\frac{1}{(y-N)^{4-2\alpha}}\\
&\le 2(2N+1)^2\int_{2N}^{\infty}\frac{1}{(v-N)^{4-2\alpha}}\d v\\
&\le \frac{18}{3-2\alpha}N^{2\alpha-1}.
\end{align*}
Thus,
\be \label{boundAn}
A_N\le \sqrt{\frac{18}{3-2\alpha}}N^{\alpha-1/2}.
\ee

For $\delta_N=1$ and $k>1+(\alpha-1/2)/(1-\alpha)$, one can  check that
\be  \label{speed}
\lim_{N\to \infty}\frac{A_N}{\delta^k_N \tau^{k-1}_N}=0
\ee

Therefore, by Proposition \ref{LLT}
$$
\limsup_{N\to \infty}A_N \mathbb{P}\left( \left| W_N^{\eta} \right|\le \frac{1}{2 \beta}\ln\left( \frac{2-\tau}{\tau} \right) \right) \le \frac{1}{\beta} \ln \left( \frac{2-\tau}{\tau} \right)<\infty,
$$
as we wanted.
\end{proof}

\begin{remark} \label{Remark2} Notice that the estimate we obtain by choosing a fixed, $N$-independent, $\delta$ as above leads to a local limit scaling, which indeed also works for independent variables. For the full model this is too much too ask, but luckily,  it is also more than we need. If we choose  an increasing sequence of $\delta_N$'s and adapt the   $\tau_N$ (which will be smaller then) we will get estimates for the probability that $|W_N^{\eta}|$ is in an interval of size $O(N^{\delta})$, which will 
be of  order $N^{ \frac{1}{2} - \alpha + \epsilon}$, for some $\epsilon$. 
Notice that $A_N$ multiplies  the Gaussian integral which always has about the same value, as long as $\tau_N$ is larger than $N^{ - (\alpha - \frac{1}{2})}$.
The ordinary central limit scaling would only require $\delta_N$ to be of order $N^{-\frac{1}{2} + \alpha}$ and the integration interval size $\tau_N$ of order $N^{ - (\alpha - \frac{1}{2})}$.  We will thus have to choose  powers which are intermediate, with a $\delta_N$ close, to  but slightly less  than,  the central limit 
scaling.
This weak local limit behavior  is what we will prove for the full model later on. Notice also,  by the way, that trivially the probability of the boundary (free) energy to be in a finite interval around zero is less than being in such an increasing sequence of intervals. 
\end{remark}

\section{Contours and extremal decomposition}

\subsection{Triangles and contours descriptions}
The first proof of the phase transition at low temperatures by contour arguments for the Dyson model is due to \cite{FrSp} for $\alpha=0$. 
Afterwards, \cite{CFMP} adapted the notion of contours for any $\alpha \in [0,\alpha_+)$, where $\alpha_+=\frac{\log 3}{\log 2}-1$, assuming that the interaction of the nearest-neighbor $J$ is large enough to show the phase transition. \cite{LP} showed that this definition of configuration of contours cannot be extended for $\alpha\in [\alpha_+,1)$. Restricting the interval to $[0,\alpha^*)$, where $\alpha^*<\alpha_+$ is a real number satisfying $\sum_{k=1}^{\infty}\frac{1}{k^{2-\alpha^*}}=2,$ \cite{BEEKR} proved that it is possible to replace $J=1$ to show the phase transition by using contour arguments. One may see some applications using the configuration of contours in \cite{BEEL,CMP17,CMPR,COP1,COP,Littin}.

In this section, we adapt to free boundary conditions the context of Littin-Picco \cite{LP} to provide  quasiadditive bounds on Hamiltonians valid for any $\alpha \in (0,1)$. This quasiadditivity will allow us to introduce the free b.c. partition function as a gas of polymers with convergent cluster expansion at low temperatures, by adapting the context (and corresponding weights) to our new notion of boundary triangles forming  contours. 

We start noticing that the Hamiltonian can be rewritten as
$$
H^{\eta}_{\Lambda}(\sigma)=\frac{\beta}{2} \sum_{x,y \in \Lambda}J_{xy}\mathbbm{1}_{\sigma_x\neq \sigma_y}  + \beta\sum_{\substack{x\in \Lambda \\ y\notin \Lambda}} \eta_yJ_{xy} {\mathbbm{1}_{\sigma_y\neq 1}}.
$$

The Gibbs measures at inverse temperature $\beta$, in finite volume $\Lambda$,  with boundary condition $\eta$, can be written as
\be \label{mueta}
\mu^{\eta}_{\Lambda}(\sigma) = 
\frac{1}{Z^{\eta}_{\Lambda}} \exp\left( -\frac{\beta}{2} \sum_{x,y \in \Lambda}J_{xy}\mathbbm{1}_{\sigma_x\neq \sigma_y}  - \beta\sum_{\substack{x\in \Lambda \\ y\notin \Lambda}} \eta_yJ_{xy} {\mathbbm{1}_{\sigma_x\neq 1}} \right),
\ee
where the partition function
\be \label{PartEta}
Z^{\eta}_{\Lambda} = \sum_{\sigma \in \Omega_{\Lambda}}\exp\left( -\frac{\beta}{2} \sum_{x,y \in \Lambda}J_{xy}\mathbbm{1}_{\sigma_i\neq \sigma_y}  - \beta\sum_{\substack{x\in \Lambda \\ y\notin \Lambda}} \eta_y J_{xy} \mathbbm{1}_{\sigma_x\neq 1}  \right)
\ee
is denoted with the same notation. For $x\in \Lambda$, define
$$
{h^{\eta}_{\Lambda}(x) = \sum_{y\notin \Lambda}J_{xy}\eta_y,}
$$

Rewrite the Hamiltonian with free boundary condition $H^f_{\Lambda}$  on $\Omega_\Lambda$ (or on $\Omega$) as
$$
H_{\Lambda}^f(\sigma) = \frac{1}{2} \sum_{x,y\in \Lambda}J_{xy}\mathbbm{1}_{\sigma_x\neq \sigma_y} .
$$
 In that case, the partition function (\ref{PartEta}) is written
$$
Z^{\eta}_{\Lambda} = \sum_{\sigma \in \Omega_{\Lambda}}\exp\left( -\beta H_{\Lambda}^f(\sigma) - \beta\sum_{x\in \Lambda} h_{\Lambda}^{\eta}(x) \mathbbm{1}_{\sigma_x\neq 1}  \right)
$$

Consider $N\ge 1$, and consider an interval $\Lambda=[-N,N]$ (All the definitions below also work for $\Lambda=[1,N]$). 
Define the dual lattice of $\Lambda$ given by
$\Lambda^*=\left\{ x+\frac{1}{2}: -N\le x\le N \right\}.$
Given a configuration $\sigma\in \Omega_{\Lambda}$, a \emph{spin-flip point} is a site $x$ in $\Lambda^*$ such that $\sigma_{x-\frac{1}{2}}\neq \sigma_{x+\frac{1}{2}}$. For each spin-flip point $x$, let us consider the interval $\left[x-\frac{1}{100},x+\frac{1}{100}\right]\subset \mathbb{R}$ and choose a real number $r_x$ in this interval such that, for every four distinct $r_{x_l}$ with $l=1,\ldots,4$, we have $|r_{x_1}-r_{x_2}|\neq |r_{x_3}-r_{x_4}|$.

For each spin-flip point $x$, we start growing a ``$\lor$-line'' at $r_x$ where this $\lor$-line is embedded in $\mathbb{R}^2$ with angles $\pi/4$ and $3\pi/4$. If at some time two $\lor$-lines starting from different spin-flip points touch -forming a triangle-, the other two lines starting from those two spin-flip points stop growing and we remove those lines -which did not form a triangle-, and  we keep continuing this process. Moreover, if one of the lines of the $\lor$-line touches the perpendicular line (boundary) at $-N$ or $N$, we also remove the other line.  
We call those triangles that touch the boundary,  which exist with free b.c. only,  \emph{boundary triangles}. A triangle $T$ is called a left (resp. right) boundary triangle  if $T$ touches the perpendicular line at $-N$ (resp. $N$).

We denote a (possibly boundary) triangle by $T$, and introduce the following notations:
$$
\begin{aligned}
\{x_{-}(T),x_{+}(T)\}&=\text{ the left and right root of the associated $\lor$-lines, respectively},\\
x_-(T)&=-N (\text{resp. }x_+(T)=N) \text{ if $T$ is a left (resp. right) boundary triangle},\\
\Delta(T)&=[x_{-}(T),x_{+}(T)]\cap \mathbb{Z}, \text{ is the base of the triangle }T,\\
|T|&=|\Delta(T)|, \text{ be the mass of the triangle }T,\\
\sf^*(T)&=\left\{ \inf \Delta(T)-\frac{1}{2},\sup \Delta(T) +\frac{1}{2} \right\},\\
\Ext(T)&=\Lambda\setminus \Delta(T) \text{ the exterior of }T.
\end{aligned}
$$ 
Equip also $\mathbb{Z}$  with the natural distance and define
$$
\d(T,T')=\d\{\sf^*(T),\sf^*(T')\}.
$$
By definition of the triangles, for every pair of triangles $T\neq T'$,
\begin{equation}\label{eq:triangle_condition}
\d(T,T')\ge \min\{|T|,|T'|\}.
\end{equation}

We denote by $\mathcal{T}_{\Lambda}$ be the set of configurations of triangles $\underline{T}=\{T_1,\ldots,T_n\}$ satisfying (\ref{eq:triangle_condition}) and such that $\Delta(T_i)\subset \Lambda$ for every $i=1\ldots,n$. Given $\sigma\in \Omega_{\Lambda}$, we denote by $\underline{T}(\sigma)$ the configuration of triangles constructed from the configuration $\sigma$.

The exterior of a configuration of triangles $\un{T}$ is defined by
$$
\Ext(\un{T})=\bigcap_{T\in \un{T}}\Ext(T).
$$

A family of triangles $\underline{T}$ is a \emph{compatible family} if there exists a configuration $\sigma\in\Omega_{\Lambda}$ such that $\underline{T}=\underline{T}(\sigma)$.

The definition of contours below is slightly different from the definition in \cite{CFMP, CMPR} because of the existence of boundary triangles.

\begin{definition}\label{def:contour}
Let $\underline{T}\in \mathcal{T}_{\Lambda}$ be a compatible configuration of triangles. A \emph{configuration of contours} $\underline{\Gamma}\equiv \underline{\Gamma}(\underline{T})$ is a partition of $\underline{T}$ whose elements, called \emph{contours}, are determined by the following properties:

\textbf{P.0:} Let $\underline{\Gamma}\equiv (\Gamma_1,\ldots,\Gamma_N)$ and $\Gamma_i=\{T_{m,i}:1\le m\le k_i\}$, thus
$$
\underline{T}=\{T_{m,i}:1\le m\le k_i,1\le i\le N\}.
$$

\textbf{P.1:} Contours are well-separated from each other. Consider the base of a contour $\Gamma$ by
$$
\Delta(\Gamma)=\bigcup_{T\in \Gamma}\Delta(T).
$$ 
Any pair $\Gamma\neq \Gamma'$ in $\underline{\Gamma}$ verifies one of the following two alternatives:

\textnormal{(1)} $\Delta(\Gamma_i)\cap \Delta(\Gamma_j)=\emptyset$.

\textnormal{(2)} Either $\Delta(\Gamma_i)\subseteq \Delta(\Gamma_j)$ or $\Delta(\Gamma_j)\subseteq \Delta(\Gamma_i)$. Moreover, supposing that the first case is satisfied, then for any triangle $T_{m,j} \in \Gamma_j$, either $\Delta(\Gamma_i) \subseteq T_{m,j}$ or $\Delta(\Gamma_i)\cap T_{m,j}=\emptyset$.

Write $|\Gamma|=\sum_{T\in \Gamma}|T|$. Then, in both cases (1) and (2), 
\be\label{compatibility2}
\d(\Gamma,\Gamma'):=\min_{\substack{T\in \Gamma \\ T'\in\Gamma'}}\d(T,T')>c\min\{|\Gamma|,|\Gamma'|\}^3,
\ee
where $c>1$ be a positive real number satisfying, as in \cite{CFMP},
\be \label{c}
\sum_{M=1}^{\infty}\frac{4M}{\lfloor cM^3 \rfloor}\le \frac{1}{2}.
\ee

\textbf{P.2:} Independence. Let $\{\underline{T}^{(1)},\ldots, \underline{T}^{(k)}\}$ be configurations of triangles; consider the contours of the configuration $\underline{T}^{(i)}$ by $\underline{\Gamma}(\underline{T}^{(i)})= \{\Gamma_j^{(i)}: j=1,\ldots,n_i\}$. If for any distinct pair $\Gamma_{j}^{(i)}$ and $\Gamma_{j'}^{(i')}$ the property P.1 is satisfied, then
$$
\underline{\Gamma}\left(\underline{T}^{(1)},\underline{T}^{(2)},\ldots, \underline{T}^{(k)}\right) = \left\{ \Gamma_j^{(i)}: j=1,\ldots, n_i;\ i=1,\ldots,k \right\}.
$$
\end{definition}

 A contour that contains a boundary triangle is called \emph{boundary contour}.

The proof of existence and uniqueness of an algorithm that produces $\underline{\Gamma}$ satisfying Definition~\ref{def:contour} is the same as in \cite{CFMP} (the presence or absence of boundary contours does not play a role here). Thanks to these boundary contours with free boundary condition, we can decompose the configuration space as $\Omega_{\Lambda}=\Omega_{\Lambda}^+ \cup \Omega_{\Lambda}^-$, where
$$
\Omega_{\Lambda}^{\pm} = \{\sigma \in \Omega_{\Lambda}: \sigma|_{\Ext(\un{T}(\sigma))} = \pm 1\}.
$$
By spin-flip symmetry, $\Omega_{\Lambda}^{\pm} = -\Omega^{\mp}_{\Lambda}$.

Define $\mathcal{T}_{\Lambda}^\pm$ to be the set of all triangle configurations $\un{T} \in \mathcal{T}_{\Lambda}$ compatible with a configuration $\sigma \in \Omega_{\Lambda}^{\pm}$. There is a bijection $\Omega^{\pm}_{\Lambda} \to \mathcal{T}_{\Lambda}$, which implies a bijection $\underline{T}^\pm \mapsto \underline{\Gamma}(\underline{T})$. Thus, we have a bijection between spin configurations in $\Omega^{\pm}_{\Lambda}$ and contour configurations.

\begin{definition}
We say that a family of contours $\{\Gamma_0,\Gamma_1,\ldots,\Gamma_n\}$ is \emph{compatible} if there exists $\sigma_{\Lambda}\in \Omega_{\Lambda}$ such that $\underline{\Gamma}=\underline{\Gamma}(\underline{T}(\sigma_{\Lambda}))$. In particular, we denote by $\Gamma \sim \Gamma'$ when $\{\Gamma,\Gamma'\}$ is compatible.
\end{definition}

For two contours $\Gamma$ and $\Gamma'$, let us define the following notations,
\begin{align*}
x_-(\Gamma) &= \min_{T\in \Gamma} x_-(T)\\
x_+(\Gamma) &= \max_{T\in \Gamma} x_+(T)\\
\supp(\Gamma) &= [x_-(\Gamma),x_+(\Gamma)]\\
\sf^*(\Gamma) &= \bigcup_{T\in \Gamma} \sf^*(T)
\end{align*}

\subsection{Approximate extremal decomposition at finite volume}

Define
$$
W_{\Lambda}^{\eta}(\sigma) = \sum_{x\in \Lambda}\sum_{y\notin \Lambda}J_{xy}\sigma_x\eta_y
$$
and the partition functions $\calz^{+,\eta}_{\Lambda}$ and $\calz^{-,\eta}_{\Lambda}$ by
\be \label{newpartition}
\calz^{\pm,\eta}_{\Lambda} = \sum_{\sigma \in \Omega^\pm_{\Lambda}}e^{-\beta H^{\eta}_{\Lambda}(\sigma)}.
\ee
Note that (\ref{newpartition}) can then be written as sum over compatible contours
$$
\calz^{\pm,\eta}_{\Lambda}=\sum_{\underline{\Gamma} \in \mathcal{O}^\pm_{\Lambda}} e^{-\beta H^{\pm,\eta}_{\Lambda}[\underline{\Gamma}]},
$$
where $\mathcal{O}^\pm_{\Lambda}$ is the set of compatible contours associated to $\Omega^\pm_{\Lambda}$. Note that, by spin-flip symmetry,
$$ 
\calz^{-,\eta}_{\Lambda}=\calz^{+,-\eta}_{\Lambda}.
$$
Thus, we can use the shorthand notations $\calz^{\eta}_{\Lambda}=\calz^{+,\eta}_{\Lambda}=\calz^{-,-\eta}_{\Lambda}$. Moreover, note that $W^{-\eta}_{\Lambda}(\sigma)=-W^{\eta}_{\Lambda}(\sigma)$ and $h^{-\eta}_{\Lambda}(x)=-h^{\eta}_{\Lambda}(x)$ for all $\eta,\sigma\in \Omega$.
One can also write
$$
H^{\eta}_{\Lambda}[\underline{\Gamma}] = \sum_{x,y\in \Lambda}J_{xy}\mathbbm{1}_{\sigma_x(\un{\Gamma})\neq \sigma_y(\un{\Gamma})} -\sum_{x\in \Lambda}h_{\Lambda}^{\eta}(x)\sigma_x
$$
for all contours $\un{\Gamma}\in \mathcal{O}^+_{\Lambda}$.
Here, $\sigma_i(\un{\Gamma}) \in \{-1,1\}$ is the value of the spin at $i$ when the contour configuration $\un{\Gamma}$ is present. Using then the decomposition
$$
Z_\Lambda^\eta= \calz^{\eta}_{\Lambda} + \calz^{-\eta}_{\Lambda}
$$
the Gibbs measure $\mu^{\eta}_{\Lambda}$ can be represented as a convex combination of two constrained Gibbs measures:
\be \label{decompo}
\mu^{\eta}_{\Lambda}(\sigma) = \left[ 1+ \frac{\calz^{-\eta}_{\Lambda}}{\calz^{\eta}_{\Lambda}} \right]^{-1} \nu^{+,\eta}_{\Lambda}(\sigma)+
 \left[1+  \frac{\calz^{\eta}_{\Lambda}}{\calz^{-\eta}_{\Lambda}} \right]^{-1} \nu^{-,\eta}_{\Lambda}(\sigma),
\ee
where the measures $\nu^{\pm,\eta}_{\Lambda}$ constrained to $\Omega^{\pm}_{\Lambda}$ are defined by
$$
\nu^{\pm,\eta}_{\Lambda}(\sigma) = \frac{1}{\calz^{\pm}_{\Lambda}} e^{-\beta H^{\eta}_{\Lambda}(\sigma)}\mathbbm{1}_{\{\sigma \in \Omega^{\pm}_{\Lambda}\}}.
$$
Define $\nu^{\eta}_{\Lambda}(\sigma)=\nu^{+,\eta}_{\Lambda}(\sigma)=\nu^{-,-\eta}_{\Lambda}(-\sigma)$.

Using $\sigma_x = 1-2\mathbbm{1}_{\sigma_x=-1}$, the partition function $\calz^{\eta}_{\Lambda}$ can be written as
$$
\calz^{\eta}_{\Lambda}
= e^{\beta W^{\eta}_{\Lambda}}\sum_{\sigma \in \Omega^+_{\Lambda}} \exp\lf -\beta \sum_{x,y\in \Lambda}J_{xy}\mathbbm{1}_{\sigma_x\neq \sigma_y} - 2\beta\sum_{x\in \Lambda}h^{\eta}_{\Lambda}(x)\mathbbm{1}_{\sigma_x=-1}\ri.
$$

\subsection{From finite-volume decomposition to metastate behavior via cluster expansion and WLLT}

Define
$$
\Xi^{\eta}_{\Lambda}=\sum_{\sigma \in \Omega^+_{\Lambda}} \exp\lf -\beta H^{f}_{\Lambda}(\sigma) - 2\beta E^{\eta}_{\Lambda}(\sigma)\ri
$$
where
$$ 
E^{\eta}_{\Lambda}(\sigma)=\sum_{x\in \Lambda}h^{\eta}_{\Lambda}(x)\mathbbm{1}_{\sigma_x=- 1}.
$$
Note that $E^{-\eta}_{\Lambda}(\sigma)=-E^{\eta}_{\Lambda}(\sigma)$. 

\begin{proposition}\label{prop_convergence_nu}
For every $1/2<\alpha<1$, there exists $\beta_3=\beta_3(\alpha)$ such that, for every $\beta>\beta_3$ and $X\Subset \Z$,
$$ 
\lim_{N\to \infty}\lVert \nu^{\eta}_{\Lambda_N} - \mu^+ \rVert_X =0.
$$
\end{proposition}

This proposition will be proven in Section \ref{sec:proofs}, after the (longer) one of Proposition 1.

Note then that, for a fixed $\tau>0$, $\ell<\tau$, $\beta>\beta_3$, and $X\Subset \Z$, there exists $N_0$ such that, for every $N\ge N_0$,
$$
\lVert \nu^{\eta}_{\Lambda_N} - \mu^+ \rVert_X<\ell.
$$
Define
\be\label{weights2}
\lambda^{\eta}_{\Lambda}=\left[ 1+\frac{\calz^{\eta}_{\Lambda}}{\calz^{-\eta}_{\Lambda}} \right]^{-1}.
\ee
Thus,
\begin{align*}
&\mathbb{P}\left( \lVert \mu^{\eta}_{\Lambda_N}-\mu^+ \rVert_X \land \lVert \mu^{\eta}_{\Lambda_N}-\mu^- \rVert_X \ge \tau\right)
\le \mathbb{P}\left( \lVert \mu^{\eta}_{\Lambda_N}-\nu^{\eta}_{\Lambda_N} \rVert_X \land \lVert \mu^{\eta}_{\Lambda_N}-\nu^{-\eta}_{\Lambda_N} \rVert_X \ge \tau-\ell\right)\\
&\le \mathbb{P}\left( (1-\lambda^{\eta}_{\Lambda_N})\lVert \nu^{\eta}_{\Lambda_N}-\nu^{-\eta}_{\Lambda_N} \rVert_X \land \lambda^{\eta}_{\Lambda_N}\lVert \nu^{\eta}_{\Lambda_N}-\nu^{-\eta}_{\Lambda_N} \rVert_X \ge \tau-\ell\right)\\
&\le \mathbb{P}\left(  \left[ \frac{\lVert \nu^{\eta}_{\Lambda_N}-\nu^{-\eta}_{\Lambda_N} \rVert_X}{\tau-\ell}-1 \right]^{-1} \le \frac{\calz^{\eta}_{\Lambda_N}}{\calz^{-\eta}_{\Lambda_N}} \le  \frac{\lVert \nu^{\eta}_{\Lambda_N}-\nu^{-\eta}_{\Lambda_N} \rVert_X}{\tau-\ell}-1 \right)\\
&= \mathbb{P}\left( |F^{\eta}_N| \le \frac{1}{2\beta}\log \left( \frac{\lVert \nu^{\eta}_{\Lambda_N}-\nu^{-\eta}_{\Lambda_N} \rVert_X}{\tau-\ell}-1\right) \right)
\end{align*}
where $F^{\eta}_N$ is the free energy given by
$$
F^{\eta}_N=\frac{1}{2\beta}(\log \calz^{\eta}_{\Lambda_N} - \log \calz^{-\eta}_{\Lambda_N}) = W^{\eta}_{\Lambda_N}+\frac{1}{2\beta}(\log \Xi^{\eta}_{\Lambda_N} - \log \Xi^{-\eta}_{\Lambda_N}).
$$
Note that
$$
\mu^{\eta}_{\Lambda}=\frac{e^{\beta F^{\eta}_{N}}}{e^{\beta F^{\eta}_{N}}+e^{-\beta F^{\eta}_{N}}}\nu^{\eta}_{\Lambda_N}+\frac{e^{-\beta F^{\eta}_{N}}}{e^{\beta F^{\eta}_{N}}+e^{-\beta F^{\eta}_{N}}}\nu^{-\eta}_{\Lambda_N}.
$$

Thus, Equation (\ref{eq:mainprop}) is equivalent to
\begin{equation}\label{FNRJDecay}
\limsup_{N\to \infty}N^{\alpha-\frac{1}{2}-\varepsilon}\mathbb{P}\left( |F^{\eta}_N| \le \frac{1}{2\beta}\log \left( \frac{\lVert \nu^{\eta}_{\Lambda_N}-\nu^{-\eta}_{\Lambda_N} \rVert_X}{\tau-\ell}-1\right) \right)<\infty.
\end{equation}

Comparing to the expression (\ref{toy_w}), we want to analyse the free energy $F^{\eta}_N$.

Using the relation 
$H_\Lambda^\eta(\sigma)=H_\Lambda^{f}(\sigma) - W_\Lambda^\eta(\sigma)$,
we shall analyse $F^{\eta}_N$ by treating it as a free b.c. part plus a boundary term which is due to the random boundary condition.  This boundary condition  can equivalently  be written as a decaying random field added to the free bc Hamiltonian,  decaying from the boundary.
 The asymptotic behaviour of these random variables will be treated thanks to the derivation of cluster expansions, applied afterwards to the asymptotic control of the characteristic functions.  The additional random field part remains bounded for intermediate decays,  and the weights (\ref{weights2}) converge to some random variables, while the constrained partition functions are unbounded for slow decays so that a Local Limit Theorem for weakly dependent sequences as in  \cite{ENS2} 
 still holds (and the weights (\ref{weights2}) converge then to $\frac{1}{2}$ (by symmetry considerations), similarly to what happens for n.n. Ising models in higher dimensions \cite{ENS1,ENS2}.
 
Let us divide the volume $\Lambda_N=\Lambda_{N_k}$ into three intervals
 $$
 [-N_k,-N_k+n_k) \cup [-N_k+n_k,N_k-n_k] \cup (N_k-n_k,N_k]. 
 $$
 and let us also decompose the difference of the free energy $F^{\eta}_{N_k}$ into the following terms,
 $$ 
 F^{\eta}_{N_k} = W^{\eta}_{\Lambda_N}+\frac{1}{2\beta}(\log \tilde{\Xi}^{\eta}_{N_k,n_k} - \log \tilde{\Xi}^{-\eta}_{N_k,n_k}) +\frac{1}{2\beta}\left(\log \frac{\Xi^{\eta}_{\Lambda_N}}{\tilde{\Xi}^{\eta}_{N_k,n_k}} - \log \frac{\Xi^{-\eta}_{\Lambda_N}}{\tilde{\Xi}^{-\eta}_{N_k,n_k}}\right),
 $$
 where
 $$
 \tilde{\Xi}^{\eta}_{N_k,n_k} = \sum_{\sigma\in \Omega_{\Lambda_{N_k}}} \exp\left(  -\beta H_{\Lambda_{N_k}(\sigma)} - 2\beta \sum_{x\in \Lambda_{N_k}}\tilde{h}^{\eta}_{N_k,n_k}(x)\mathbbm{1}_{\sigma_x=-1}\right)
 $$
and
$$ 
\tilde{h}^{\eta}_{N_k,n_k}(x)=
\begin{cases}
h^{\eta}_{\Lambda_{N_k}}(x) &\text{ if }x\in [-N_k+n_k,N_k-n_k]\\
0 &\text{ if }x\in [-N_k,-N_k+n_k) \cup (N_k-n_k,N_k].
\end{cases}
$$ 
Note that
 $$
\left|\log \frac{\Xi^{\eta}_{\Lambda_N}}{\tilde{\Xi}^{\eta}_{N_k,n_k}} \right| \le 4\beta n_k \le 8\beta N_k^{\frac{1}{2}\left( \alpha-\frac{1}{2} \right)}
 $$
 Because all $\tilde{h}^{\eta}_{N_k,n_k}(x)$ in the central interval are small enough we can use the cluster expansion for the truncated model. The full model needs adding the fields in the left and right intervals, but those contribute at most a term of order $n$ which is small in comparison with $N_k^{\alpha -1/2}$ which is the order of the size of the (free) energy of all the boundary terms together. 

So the conclusion of the WLLT for the truncated model still applies for the untruncated model, if we want to say that with large probability we avoid a window around zero of size $N_k^{\alpha - 3/2 - \delta}$, as long as this window is large compared to $n_k$.


\section{Cluster Expansion at low temperatures}

To control the weights entering into the decomposition of the random Gibbs measures $\mu^\eta$, we use cluster expansions of partition functions with free or mixed (+,$\eta$) boundary conditions, either directly on the constrained partition functions $\mathcal{Z}^{\pm\eta}_\Lambda$, or starting from the former (free b.c.) which will allow us to treat the random b.c. {\em via} a ``separate'', typically decaying, random field. 

{We restricted the results in this section to $1/2<\alpha<1$ because of the choice of the parameter $n_k<N_k^{\frac{1}{2}\left(\alpha-\frac{1}{2}\right)}$. However, all the results holds for any $0<\alpha<1$ if we replace the condition of $n_k$ by $n_k=o(N_k)$.}

In all cases, we need to adapt the existing cluster expansions for $+$ (and $\pm$) b.c. from the papers \cite{CFMP,CMPR,COP1} (restricted to some values of $\alpha$), and more generally,  Littin {\em et al.} \cite{LP} (valid for any $\alpha \in (0,1)$ and this amounts to first also adapting the existing contour descriptions by introducing adapted boundary triangles and contours\footnote{Specific to free b.c., in addition to classical triangles and contours of the above references.}. based on the same spin-flip  points, to partition the finite-volume configuration space with free b.c. as
$$
\Omega_\Lambda= \Omega_\Lambda^{+} \cup \Omega_\Lambda^{-} 
$$

Given a collection of contours $\underline{\Gamma}$, a contour $\Gamma \in \un{\Gamma}$ is \emph{external} with respect to $\un{\Gamma}$ if each triangle of $\Delta(\Gamma)$ is not contained in some triangle that belongs to the other contours of $\un{\Gamma}$. A contour which is not external is called \emph{internal}.

For a collection of contours $\un{\Gamma}$,  let $\Gamma^{\ex}$  be an external contour of $\un{\Gamma}$, and let  $n^{\ex}(\un{\Gamma}) =|\un{\Gamma}|$ be the number of external contours in $\un{\Gamma}$.  Let  $\un{\Gamma}^{\inte}(\Gamma)$  be a set of contours internal to $\Gamma$ in $\un{\Gamma}$, and let  $n^{\inte}(\Gamma) = |\un{\Gamma}^{\inte}(\Gamma)|$ be the cardinality of $\un{\Gamma}^{\inte}(\Gamma)$.

A collection of contours $\un{\Gamma}$ can be uniquely decomposed as
$$
\un{\Gamma}=\bigcup_{j=1}^{n^{\ex}(\un{\Gamma})} \left( \Gamma_j^{\ex} \cup \un{\Gamma}^{\inte}(\Gamma^{\ex}_j) \right).
$$

A \emph{polymer} $R$ is a set of mutually external compatible contours. Denote
$$
\Delta(R)=\bigcup_{\Gamma \in R}\Delta(\Gamma).
$$
Two polymers $R$ and $R'$ are \emph{compatible}, denoted by $R\sim R'$, if the following statements hold.
\begin{enumerate}
\item[(i)] For all $\Gamma \in R$ and $\Gamma' \in R'$, the contours $\Gamma \sim \Gamma'$ are compatible.
\item[(ii)] One of the following condition holds:
$$
\Delta(R) \cap \Delta(R')=\emptyset
$$
or
$$
\sum_{\Gamma \in R} \mathbbm{1}_{\Delta(R')\subset \Delta(\Gamma)}+\sum_{\Gamma' \in R'} \mathbbm{1}_{\Delta(')\subset \Delta(\Gamma')} =1.
$$
\end{enumerate}

Define
$$ 
\tilde{E}^{\eta}_{N_k,n_k}(\bar{\Gamma})=\sum_{x\in \bar{\Gamma}}\tilde{h}^{\eta}_{N_k,n_k}(x)\mathbbm{1}_{\sigma_x=-1}.
$$

Define the activity function $\xi_{N_k,n_k}^{\eta}$ by
$$
 \xi_{N_k,n_k}^{\eta}(R)=
\begin{cases}
A^{\eta}_{N_k,n_k}(\Gamma)w^{\eta}_{N_k,n_k}(\Gamma)& \text{ if }R=\{\Gamma\},\\
A^{\eta}_{N_k,n_k}(R)\prod_{\Gamma \in R}w^{\eta}_{N_k,n_k}(\Gamma)& \text{ otherwise},
\end{cases}
$$
where the functions $A^{\eta}_{\Lambda}$ and $w^{\eta}_{\Lambda}$ are defined as follows
$$
A^{\eta}_{N_k,n_k}(R)=\frac
{
\displaystyle\sum_{\substack{\underline{\Gamma}^{\inte}(\Gamma) \\ \Gamma \in R}} \left[ \prod_{\Gamma \in R} \exp\left(-\beta H^f_{\Lambda_{N_k}}[\Gamma \cup \underline{\Gamma}^{\inte}(\Gamma)]- 2\beta \tilde{E}^{\eta}_{N_k,n_k}(\Gamma\cup \underline{\Gamma}^{\inte}(\Gamma) ) \right) \right]\varphi(R)
}
{
 \displaystyle\sum_{\substack{\underline{\Gamma}^{\inte}(\Gamma) \\ \Gamma \in R}} \prod_{\Gamma \in R} \exp\left(-\beta H^f_{\Lambda_{N_k}}[\Gamma \cup \underline{\Gamma}^{\inte}(\Gamma)] - 2\beta\tilde{E}^{\eta}_{N_k,n_k}(\Gamma\cup \underline{\Gamma}^{\inte}(\Gamma) )  \right)
},
$$
where, after introducing the usual connected graph $\mathcal{G}_\cdot$, 
$$ 
 \varphi(R) = 
 \begin{cases}
 1 &\text{ if }|R|=1,\\
 \displaystyle\sum_{g\in \mathcal{G}_{|R|}}\prod_{ij\in E(g)}f_{ij} &\text{ otherwise}
 \end{cases}
$$
 and, using the interaction $\calk$ between two contours
$$
 f_{ij} = \exp(\beta \calk[\Gamma^*_i, \Gamma^*_j])-1. 
$$
 The function $w^{\eta}_{N_k,n_k}$ is defined by
$$
 w^{\eta}_{N_k,n_k}(\Gamma)= \frac{\displaystyle\sum_{\un{\Gamma}^{\inte}(\Gamma)\sim \Gamma}\exp\left(-\beta H_{\Lambda_{N_k}}^{f}[\Gamma \cup \un{\Gamma}^{\inte}(\Gamma)]- 2\beta \tilde{E}^{\eta}_{N_k,n_k}(\Gamma\cup \underline{\Gamma}^{\inte}(\Gamma)) \right)}{\displaystyle\sum_{\un{\Gamma}^{\inte}(\Gamma)\sim \Gamma} \exp\left(-\beta H_{\Lambda_{N_k}}^{f}[\un{\Gamma}^{\inte}(\Gamma)]- 2\beta
 \tilde{E}^{\eta}_{N_k,n_k}(\underline{\Gamma}^{\inte}(\Gamma)) \right)}.
$$

We can adapt the Quasi-Additive decomposition of the Hamiltonian from \cite{LP} if we include the boundary contours with free boundary condition.
\begin{theorem}
Let $\{\Gamma_0\} \cup \un{\Gamma}$ with $\un{\Gamma}=\{\Gamma_1,\ldots,\Gamma_n\}$ be an arbitrary configuration of compatible contours. Then, for all $0\le \alpha <1$ the following inequalities are satisfied:
\begin{align*}
H^f_{\Lambda_{N}}[\{\Gamma_0\} \cup \un{\Gamma}] -H^f_{\Lambda_N}[ \un{\Gamma}] &\ge K_c(\alpha)H^f_{\Lambda_N}[\Gamma_0],\\
H^f_{\Lambda_N}[\Gamma_0,\Gamma_1,\ldots,\Gamma_n] &\ge K_c(\alpha)\sum_{i=0}^nH^f_{\Lambda_N}[\Gamma_i],
\end{align*}
where 
$$
K_c(\alpha)=1-\frac{2-\alpha}{c^{1-\alpha}}+\frac{\pi^2}{6c}
$$
satisfies $\frac{1}{2}<K_c(\alpha)<1$ for sufficiently large $c$, the constant in  (\ref{c}).
\end{theorem}

\begin{theorem}
The partition function $\tilde{\Xi}^{\eta}_{N_k,n_k}$ can be written in terms of a gas of polymers with activities $\xi^{\eta}_{N_k,n_k}$ as follows:
$$
\tilde{\Xi}^{\eta}_{N_k,n_k} = 1+\sum_{\underline{R}\neq \emptyset} \prod_{R\in \underline{R}} \xi_{N_k,n_k}^{\eta}(R)\prod_{(R,R')}\mathbbm{1}_{R\sim R'},
$$
where the sum is over $\mathcal{R}_{\Lambda}$, the set of polymers with support in $\Lambda$ with $\Delta(R)\cap \mathbb{Z}\subset \Lambda$ for all $R\in \un{R}$.
\end{theorem}


For $n\ge 1$, define the Ursell function $\phi^T(R_1,\ldots,R_n)$ by
$$
\phi^T(R_1,\ldots, R_n) =
\begin{cases}
1 & n=1\\
\displaystyle\sum_{\substack{G \subset \mathcal{G}_{\{R_1,\ldots,R_n\}} \\ G \ \text{conn. spann.}}}(-1)^{e(G)} & n\ge 2,\ \mathcal{G}_{\{R_1,\ldots,R_n\}}\text{ connected}\\
0 & n\ge 2,\ \mathcal{G}_{\{R_1,\ldots,R_n\}}\text{ not connected}
\end{cases}
$$
where $\mathcal{G}_{\{R_1,\ldots,R_n\}}$ is the graph of vertices $\{1,\ldots,n \}$ and edges
\be\label{edges}
\{\{i,j\}: R_i\sim R_j, 1\le i,j\le n, i\neq j\},
\ee
and $G$ ranges over all its connected spanning subgraphs. We denote $e(G)$ be the number of edges in $G$.

We know that the polymer expansion can be formally written as an exponential of a series, that is,
\be \label{6-40}
\log \tilde{\Xi}^{\eta}_{N_k,n_k} = \sum_{n=1}^{\infty}\frac{1}{n!}\sum_{R_1}\ldots \sum_{R_n}\phi^T(R_1,\ldots,R_n)\prod_{i=1}^n \xi^{\eta}_{N_k,n_k}(R_i).
\ee
Note that we need extra conditions to check that the series is absolutely convergent. All the results stated below will assist us to prove the convergence for all $\alpha\in [0,1)$ and for sufficiently small temperatures.


\begin{theorem}
For every triangle configuration $\underline{T}$ and any contour $\Gamma_0\in \underline{\Gamma}(\underline{T})$, we have
$$
\sum_{x\in \underline{T}}\tilde{h}^{\eta}_{N_k,n_k}(x)\mathbbm{1}_{\sigma_x=-1} - \sum_{x\in \underline{T}\setminus \Gamma_0}\tilde{h}^{\eta}_{N_k,n_k}(x)\mathbbm{1}_{\sigma_x=-1}  \ge -\sum_{T\in \Gamma_0}\sum_{x\in T}|\tilde{h}^{\eta}_{N_k,n_k}(x)|.
$$
\end{theorem}


The proof is very similar to the computation in the proof of Theorem 2 in \cite{BEEKR} (see also Theorem 1.21 in \cite{KimThes}).

\begin{proposition}\label{prop:goodeta}
For every $\frac{1}{2}<\alpha<1$ and $0<\varepsilon<1-\alpha$, there exists $\Omega_0\subseteq \Omega$  with $\mathbb{P}(\Omega_0)=1$  such that, for every $\eta\in \Omega_0$, there exists a sequence $(n_k)_{k\ge 1}$ of increasing positive integers such that $n_k<N_k^{\frac{1}{2}\left(\alpha-\frac{1}{2}\right)}$ and
$$
|\tilde{h}^{\eta}_{N_k,n_k}(x)|\le (N_k+x)^{\alpha-\frac{3}{2}+\varepsilon}+(N_k-x)^{\alpha-\frac{3}{2}+\varepsilon}
$$
for all  $x\in [-N_k+n_k, N_k-n_k]$.
\end{proposition}

\begin{proof}
Note that, for every $x\in (-N,N)$,
$$
\sum_{y\in \Lambda^c_{N}}J^2_{xy}\le \frac{1}{3-2\alpha}\lf (N+x)^{2\alpha-3} + (N-x)^{2\alpha-3} \ri.
$$
Choosing $t>0$ satisfying $t^{-1}=(N+x)^{\alpha-\frac{3}{2}}+(N-x)^{\alpha-\frac{3}{2}}$, by Chernoff's inequality,
\begin{align*}
\mathbb{P}\lf h^{\eta}_{\Lambda_{N}}(x)\ge (N+x)^{\alpha-\frac{3}{2}+\varepsilon}+(N-x)^{\alpha-\frac{3}{2}+\varepsilon}\ri
&\le e^{-(N-|x|)^{\varepsilon}} \mathbb{E}\left( e^{t \cdot h^{\eta}_{\Lambda_{N}}(x)} \right)\\
&\le e^{-(N-|x|)^{\varepsilon}}  \prod_{y\in \Lambda^c_{N}}\cosh\lf tJ_{xy} \ri\\
&\le \exp\left(-(N-|x|)^{\varepsilon}+t^2\sum_{y\in \Lambda^c_{N}}J^2_{xy}\right)\\ 
&\le Ce^{-(N-|x|)^{\varepsilon}},
\end{align*}
where $C=\exp\left( \frac{1}{3-2\alpha} \right)$. Note that
$$
e^{-x^{\varepsilon}}<x^{-1-\varepsilon} \quad \text{for every $x$ sufficiently large}.
$$
(The inequality above is true for every $x>\sqrt{\frac{1+\varepsilon}{\varepsilon}}$).

For every $k$, the sequence $N_k$ has been chosen such that 
$$
\sum_{x=m_k}^{\infty} e^{-x^{\varepsilon}}<\frac{1}{k^2}
$$
for $m_k<N_k^{\frac{1}{2}\left(\alpha-\frac{1}{2}\right)}$. Define $L_k=N_k-m_k$. We have
\begin{equation}\label{BCant}
\sum_{k=1}^{\infty}\mathbb{P}\left( \bigcup_{x=-L_k}^{L_k} \left[h^{\eta}_{\Lambda_{N_k}}(x) \ge (N_k+x)^{\alpha-\frac{3}{2}+\varepsilon}+(N_k-x)^{\alpha-\frac{3}{2}+\varepsilon}\right] \right) 
\le \sum_{k=1}^{\infty}  U_k
\end{equation}
where 
\begin{align*}
U_k&=\sum_{x=-L_k}^{L_k} \mathbb{P}\left(h^{\eta}_{\Lambda_{N_k}}(x) \ge (N_k+x)^{\alpha-\frac{3}{2}+\varepsilon}+(N_k-x)^{\alpha-\frac{3}{2}+\varepsilon} \right)\\
& \le C \sum_{x=-L_k}^{L_k}e^{-(N_k-|x|)^{\varepsilon}} \\
& \le 2C\sum_{x=0}^{L_k} e^{-(N_k-x)^{\varepsilon}} \\
&\le 2C\sum_{x=m_k}^{N_k-m_k} e^{-x^{\varepsilon}} \\
& \le 2C\frac{1}{k^2}.
\end{align*}
Thus, the series on the left-hand in (\ref{BCant}) converges. By Borel-Cantelli,
$$
\mathbb{P}\left(\displaystyle\bigcup_{x=-L_k}^{L_k} \left[h^{\eta}_{\Lambda_{N_k}}(x) \ge (N_k+x)^{\alpha-\frac{3}{2}+\varepsilon}+(N_k-x)^{\alpha-\frac{3}{2}+\varepsilon}\right]  \quad \text{i.o.} \right)=0.
$$

Therefore, there exists a subsequence $(n_k)_k$ in $(m_k)_k$ such that 
$$
h^{\eta}_{\Lambda_{N_k}}(x)<(N_k+x)^{\alpha-\frac{3}{2}+\varepsilon}+(N_k-x)^{\alpha-\frac{3}{2}+\varepsilon} \quad \text{$\bbp$-a.e.}
$$
for all $x\in [-N_k+n_k, N_k-n_k]$. For the other inequality, note that
\begin{align*}
&\mathbb{P}\lf -h^{\eta}_{\Lambda_{N_k}}(x)\ge (N_k+x)^{\alpha-\frac{3}{2}+\varepsilon}+(N_k-x)^{\alpha-\frac{3}{2}+\varepsilon}\ri \\
&= \mathbb{P}\lf h^{-\eta}_{\Lambda_{N_k}}(x)\ge (N_k+x)^{\alpha-\frac{3}{2}+\varepsilon}+(N_k-x)^{\alpha-\frac{3}{2}+\varepsilon}\ri\\
&= \mathbb{P}\lf h^{\eta}_{\Lambda_{N_k}}(x)\ge (N_k+x)^{\alpha-\frac{3}{2}+\varepsilon}+(N_k-x)^{\alpha-\frac{3}{2}+\varepsilon}\ri,
\end{align*}
as we desired.
\end{proof}

From Proposition \ref{prop:goodeta}, let us call $\Omega_0$ to be the set of \emph{good boundary conditions}. Also, let us define
$$
\tilde{h}_{N_k,n_k,\varepsilon}(x)=
\begin{cases}
(N_k+x)^{\alpha-\frac{3}{2}+\varepsilon}+(N_k-x)^{\alpha-\frac{3}{2}+\varepsilon} &\text{ if }x\in [-N_k+n_k,N_k-n_k],\\
0 &\text{ if }x=[-N_k,-N_k+n_k)\cup (N_k-n_k,N_k]
\end{cases}
$$
We will omit the dependence on $\alpha$ in the notation of $\tilde{h}_{N_k,n_k,\varepsilon}$.

Once we truncate the boundary energy at the intervals $[-N_k,-N_k+n_k]$ and $[N_k-n_k,N_k]$, the external field at $x$ is equal to zero for that interval, and we can bound $|\tilde{h}^{\eta}_{N_k,n_k}(x)|$ by $\tilde{h}_{N_k,n_k,\varepsilon}(x)$ for $x\in [-N_k+n_k,N_k-n_k]$ conditioned to the good boundary condition $\eta\in \Omega_0$.

Define a new activity function $\vartheta_{N_k,n_k,\varepsilon}$
$$
\vartheta_{N_k,n_k,\varepsilon}(\Gamma)=\exp\left(-\beta(2K_c(\alpha)-1)H^f_{\Lambda_{N_k}}[\Gamma] +2\beta \sum_{T\in \Gamma_0}\sum_{x\in T}\tilde{h}_{N_k,n_k,\varepsilon}(x)\right).
$$
Note that $\vartheta_{N_k,n_K,\varepsilon}$ is translation invariant. It is not hard to show that, under the conditions in Proposition \ref{prop:goodeta}, for any contour $\Gamma$ and for any polymer $R=\{\Gamma_1,\ldots,\Gamma_p\}$,
\begin{align*}
w^{\eta}_{N_k,n_k}(\Gamma)&\le \exp\left(-\beta K_c(\alpha) H^f_{\Lambda_{N_k}}[\Gamma] +2\beta \sum_{T\in \Gamma_0}\sum_{x\in T}\tilde{h}_{N_k,n_k,\varepsilon}(x)\right),\\
|A^{\eta}_{N_k,n_k}(R)| &\le \prod_{\Gamma\in R}e^{\beta(1-K_c(\alpha))H^f_{\Lambda_{N_k}}[\Gamma]}.
\end{align*}
The inequalities above lead to the following bound for the activity function $\xi^{\eta}_{N_k,n_k}$,
\begin{equation}\label{ineq:eta}
|\xi^{\eta}_{N_k,n_k}(R)|\le \prod_{i=1}^p \vartheta_{N_k,n_k,\varepsilon}(\Gamma_i).
\end{equation}
Define the partition function
$$
\Upsilon_{N_k,n_k,\varepsilon} = 1+\sum_{\un{\Gamma}\neq \emptyset} \prod_{\Gamma \in \un{\Gamma}}\vartheta_{N_k,n_k,\varepsilon}(\Gamma) \prod_{\{\Gamma,\Gamma'\}}\mathbbm{1}_{\Gamma \sim \Gamma'}.
$$
Using (\ref{ineq:eta}), we have the bound
$$ 
\tilde{\Xi}^{\eta}_{N_k,n_k} \le \Upsilon_{N_k,n_k,\varepsilon}
$$
that leads to
$$
\log \tilde{\Xi}^{\eta}_{N_k,n_k} \le \log \Upsilon_{N_k,n_k,\varepsilon}.
$$
Thus, to show that (\ref{6-40} (or (\ref{Exppluseta}) in Proposition 9) is absolutely convergent, it is enough to prove the absolute convergence of the series
\be\label{absconvbarxi}
\ln \Upsilon_{N_k,n_k,\varepsilon} = \sum_{n=1}^{\infty}\frac{1}{n!}\sum_{\Gamma_1}\ldots \sum_{\Gamma_n}\phi^T(\Gamma_1,\ldots,\Gamma_n)\prod_{i=1}^n \vartheta_{N_k,n_k,\varepsilon}(\Gamma_i).
\ee

The following proposition gives a lower bound for the energy of a contour. Note that the constant in the  bound is  half of the constant given in \cite{CFMP,CMPR},  since we should consider  boundary contours,  the energy of which  is half of the energy of a non-boundary contour.

For $\alpha>0$, we define
$$
\lVert \Gamma \rVert_{\alpha} = \sum_{T\in \Gamma} |T|^{\alpha}
$$
and, for $\alpha=0$, we define $\zeta_0=2$ and
$$
\lVert \Gamma \rVert_{0} = \sum_{T\in \Gamma} (4+\ln |T|).
$$
Define also
$$
\chi_{\alpha}(m) =
\begin{cases}
m^{\alpha} &\text{ if }\alpha>0\\
\ln (m) +4 &\text{ if }\alpha=0
\end{cases}
$$

\begin{proposition}\label{prop:energyhamiltonian}
If $0\le \alpha<\alpha_+=(\log 3)/(\log 2)-1$, then for any contour $\Gamma$,
$$
H^{f}_{\Lambda_N}(\Gamma)\ge \frac{\zeta_{\alpha}}{2}\lVert \Gamma \rVert_{\alpha}.
$$
\end{proposition}

\begin{proof}
If $T$ is a left boundary triangle with length $|T|=L$, we have
\begin{align*}
\mathcal{W}^{\alpha}_{\Lambda_N}(T) &= J-1+ \sum_{x=-N}^{L-N-1}\left( \sum_{y=L-N}^{2L-N-1}\frac{1}{|x-y|^{2-\alpha}} - \sum_{y=2L-N}^{N}\frac{1}{|x-y|^{\alpha}} \right)\\
&=J-1+\sum_{x=1}^{L}\left( \sum_{y=L+1}^{2L}\frac{1}{|x-y|^{2-\alpha}} - \sum_{y=2L+1}^{2N+1}\frac{1}{|x-y|^{\alpha}} \right)\\
&\ge J-1+\sum_{x=1}^{L}\left( \sum_{y=L+1}^{2L}\frac{1}{|x-y|^{2-\alpha}} - \sum_{y=2L+1}^{\infty}\frac{1}{|x-y|^{\alpha}} \right)\\
&= \mathcal{W}_{\alpha}(T),
\end{align*}

Note that, if $T$ is not a boundary triangle, then Lemma A.1 in \cite{CFMP} gives us
$$ 
\mathcal{W}_{\alpha}(T)\ge 
\begin{cases}
\zeta_{\alpha} |T|^{\alpha} &\text{ if }\alpha>0,\\
2\ln |T|+8 &\text{ if }\alpha=0.
\end{cases}
$$
If $T$ is a boundary triangle, we have
\be\label{eq:wt}
\mathcal{W}_{\alpha}(T)\ge 
\begin{cases}
\frac{\zeta_{\alpha}}{2} |T|^{\alpha} &\text{ if }\alpha>0,\\
\ln |T|+4 &\text{ if }\alpha=0.
\end{cases}
\ee
Let us use the bound (\ref{eq:wt}) for every triangle. In this case, for a contour $\Gamma$, we have
$$
\sum_{T\in \Gamma} \mathcal{W}_{\alpha}(T)\ge \frac{\zeta_{\alpha}}{2}\lVert \Gamma \rVert_{\alpha},
$$
as we desired.
\end{proof}

The following theorem is the same as Theorem 4.1 in \cite{CFMP} to bound the entropy of contours. See also Theorem 1.23 in \cite{KimThes}.

\begin{theorem}[Theorem 4.1 in \cite{CFMP}]\label{thm:entropy}
Let $\alpha \in [0,1)$ and $b$ be a positive real number. If $b$ is sufficiently large, the following inequality holds for every $m\ge 1$,
$$
\sum_{\substack{\Gamma \ni 0 \\ |\Gamma|=m}}e^{-b \lVert \Gamma \rVert_{\alpha}}\leq 2me^{-b \chi_{\alpha}(m)}.
$$
\end{theorem}

For $a>0$, define
$$
\rho_{N_k,n_k,\varepsilon}(\beta,\alpha,a)=\sum_{\Gamma \ni 0} \exp\left( -a\beta(2K_c(\alpha)-1)H^f_{\Lambda_{N_k}}[\Gamma] +2\beta \sum_{T\in \Gamma}\sum_{x\in T}\tilde{h}_{N_k,n_k,\varepsilon}(x)\right).
$$

The next theorem is the Peierls bound at sufficiently low temperature. Although we restricted ourselves here to $a=1/2$ which is the right constant to prove the absolute convergence of the cluster expansion, the proof works for all $a>0$.

\begin{theorem}\label{thm:rho}
Let $\frac{1}{2}<\alpha<1$ and $\Omega_0$ be the set of good boundary conditions. There exists a constant $C^{(2)}_{\alpha}>0$ such that, for sufficiently large $\beta$ and $k$,
$$
\rho_{N_k,n_k,\varepsilon}\left(\beta,\alpha,\frac{1}{2}\right)<e^{-\frac{\beta}{2} C^{(2)}_{\alpha}}.
$$
Moreover, $\rho_{N_k,n_k,\varepsilon}\left(\beta,\alpha,1/2\right)$ converges to zero when $\beta\to \infty$.
\end{theorem}

\begin{proof}
Let $T$ be a (boundary) triangle. By Proposition \ref{prop:goodeta}, we have a sequence $n_k<N_k^{\frac{1}{2}\left(\alpha-\frac{1}{2}\right)}$ such that, for every $\eta\in \Omega_0$ and $0<\varepsilon<1-\alpha$,
\begin{align*}
\sum_{x\in T} \tilde{h}_{N_k,n_k,\varepsilon}(x)
&\le  \sum_{x=-N_k+n_k}^{-N_k+n_k+|T|-1}\frac{1}{(N_k+x)^{\frac{3}{2}-\alpha-\varepsilon}}+ \sum_{x=N_k-n_k-|T|+1}^{N_k-n_k}\frac{1}{(N_k-x)^{\frac{3}{2}-\alpha-\varepsilon}}\\
&= 2\sum_{x=n_k}^{n_k+|T|-1}\frac{1}{x^{\frac{3}{2}-\alpha-\varepsilon}}\\
&\le  2\int_{n_k-1}^{n_k+|T|-1}\frac{1}{u^{\frac{3}{2}-\alpha-\varepsilon}}\d u \\
&=\frac{1}{\alpha-\frac{1}{2}+\varepsilon}\left( (n_k+|T|-1)^{\alpha-\frac{1}{2}+\varepsilon}-(n_k-1)^{\alpha-\frac{1}{2}+\varepsilon} \right)\\
&=\frac{|T|^{\alpha-\frac{1}{2}+\varepsilon}}{\alpha-\frac{1}{2}+\varepsilon}
\left(
\left(1+\frac{n_k-1}{|T|}\right)^{\alpha-\frac{1}{2}+\varepsilon}
 - \left(\frac{n_k-1}{|T|}\right)^{\alpha-\frac{1}{2}+\varepsilon} \right).
\end{align*}

Note that, for every $u>0$ and $0<\gamma<1$, the function $f(u)=(u+1)^{\gamma}-u^{\gamma}$ is decreasing. Also, the following inequality is true for every $u>0$ and $0<\gamma<1$,
$$ 
(u+1)^{\gamma}-u^{\gamma} \le \gamma u^{\gamma-1}.
$$
For a fixed positive number $L>1$, let us separate the cases when $|T|<L$ and $|T|\ge L$. For $|T|<L$, we have
\begin{align*}
\sum_{x\in T} \tilde{h}_{N_k,n_k,\varepsilon}(x) 
&\le \frac{|T|^{\alpha-\frac{1}{2}+\varepsilon}}{\alpha-\frac{1}{2}+\varepsilon}
\left(\left(1+\frac{n_k-1}{L}\right)^{\alpha-\frac{1}{2}+\varepsilon}
 - \left(\frac{n_k-1}{L}\right)^{\alpha-\frac{1}{2}+\varepsilon} \right)\\
 &\le |T|^{\alpha-\frac{1}{2}+\varepsilon} \left(\frac{L}{n_k-1}\right)^{\frac{3}{2}-\alpha-\varepsilon}.
\end{align*}

For $L\le |T|\le N_k$, since $f(x)=(x+1)^{\gamma}-x^{\gamma} \le 1$ for all $x>0$ and $0<\gamma<1$, we have
$$ 
\sum_{x\in T} \tilde{h}_{N_k,n_k,\varepsilon}(x)  \le \frac{|T|^{\alpha-\frac{1}{2}+\varepsilon}}{\alpha-\frac{1}{2}+\varepsilon}.
$$

Since $\alpha>1/2$ and $0<\alpha-1/2+\varepsilon<1/2$, choose $0<\alpha'<\alpha$ such that 
$$
\alpha-\frac{1}{2}+\varepsilon<\alpha'<\alpha_+=\frac{\log 3}{\log 2}-1.
$$
Writing $H^f_{\Lambda}[\Gamma]=H^f_{\Lambda,\alpha}[\Gamma]$ to keep $\alpha$ explicit in the expression, we have, by Proposition \ref{prop:energyhamiltonian},
$$ 
H^f_{\Lambda,\alpha}[\Gamma]\ge  H^f_{\Lambda,\alpha'}[\Gamma]\ge \frac{\zeta_{\alpha'}}{2}\lVert \Gamma\rVert_{\alpha'} \ge \frac{\zeta_{\alpha'}}{2}\lVert \Gamma\rVert_{\alpha-\frac{1}{2}+\varepsilon}.
$$
Thus,
\begin{align*}
&-\frac{\beta}{2}(2K_c(\alpha)-1)H^f_{\Lambda_{N_k}}[\Gamma] +2\beta \sum_{T\in \Gamma}\sum_{x\in T}\tilde{h}_{N_k,n_k,\varepsilon}(x)\\
\le & -\beta \left[ \frac{\zeta_{\alpha'}}{4}(2K_c(\alpha)-1)\sum_{\substack{T\in \Gamma \\ |T|<L}}|T|^{\alpha'} - 4\left(\frac{L}{n_k-1}\right)^{\frac{3}{2}-\alpha-\varepsilon} \sum_{\substack{T\in \Gamma \\ |T|<L}}|T|^{\alpha-\frac{1}{2}+\varepsilon} \right]\\
& -\beta \left[  \frac{\zeta_{\alpha'}}{4}(2K_c(\alpha)-1)\sum_{\substack{T\in \Gamma \\ |T|\ge L}}|T|^{\alpha'} - \frac{4}{\alpha-\frac{1}{2}+\varepsilon}\sum_{\substack{T\in \Gamma \\ |T|\ge L}}|T|^{\alpha-\frac{1}{2}+\varepsilon} \right]\\
\le & -\beta \left[ \frac{\zeta_{\alpha'}}{4}(2K_c(\alpha)-1) - 4\left(\frac{L}{n_k-1}\right)^{\frac{3}{2}-\alpha-\varepsilon} \right]\sum_{\substack{T\in \Gamma \\ |T|<L}}|T|^{\alpha-\frac{1}{2}+\varepsilon}\\
& -\beta \left[  \frac{\zeta_{\alpha'}}{4}(2K_c(\alpha)-1)\sum_{\substack{T\in \Gamma \\ |T|\ge L}}|T|^{\alpha'} - \frac{4}{\alpha-\frac{1}{2}+\varepsilon}\sum_{\substack{T\in \Gamma \\ |T|\ge L}}|T|^{\alpha-\frac{1}{2}+\varepsilon} \right]
\end{align*}
For the first term, taking $k$ large enough, we have, for a sufficiently small $\delta_0>0$,
$$
\frac{\zeta_{\alpha'}}{4}(2K_c(\alpha)-1) - 4\left(\frac{L}{n_k-1}\right)^{\frac{3}{2}-\alpha-\varepsilon} > \frac{\zeta_{\alpha'}}{4}(2K_c(\alpha)-1)  - \delta_0 >0
$$
For the second term, since
$$ 
|T|^{\alpha'}\ge L^{\alpha'-\alpha+\frac{1}{2}-\varepsilon}|T|^{\alpha-\frac{1}{2}+\varepsilon}
$$
for every triangle $T$ with $|T|\ge L$, taking $L$ sufficiently large, we have
$$
 \frac{\zeta_{\alpha'}}{4}(2K_c(\alpha)-1) L^{\alpha'-\alpha+\frac{1}{2}-\varepsilon} - \frac{4}{\alpha-\frac{1}{2}+\varepsilon}>0
$$
Define
$$
C^{(2)}_{\alpha} = \frac{\zeta_{\alpha'}}{4}(2K_c(\alpha)-1) - \delta_0 +  \frac{\zeta_{\alpha'}}{4}(2K_c(\alpha)-1) L^{\alpha'-\alpha+\frac{1}{2}-\varepsilon} -\frac{4}{\alpha-\frac{1}{2}+\varepsilon}
$$
Therefore, using Theorem \ref{thm:entropy}, for sufficiently large $\beta$,
\begin{align*}
\rho_{N_k,n_k,\varepsilon}(\beta,\alpha) &\le \sum_{\Gamma \ni 0}\exp\left(-\beta C^{(2)}_{\alpha} \lVert \Gamma \rVert_{\alpha-\frac{1}{2}+\varepsilon} \right)\\
&= \sum_{m=1}^{\infty} \sum_{\substack{\Gamma \ni 0 \\ |\Gamma|=m}}e^{-\beta C^{(2)}_{\alpha}\lVert \Gamma \rVert_{\alpha-\frac{1}{2}+\varepsilon}}\\
&\le \sum_{m=1}^{\infty}2me^{-\beta C^{(2)}_{\alpha}m^{\alpha-\frac{1}{2}+\varepsilon}}\\
& \le e^{-\frac{\beta}{2} C^{(2)}_{\alpha}}
\end{align*}
as we desired. The last inequality comes from the following argument. There exists $u_0>0$ such that $\ln u<u^{\alpha-\frac{1}{2}+\varepsilon}$ for all $u>u_0$. Thus,
\begin{align*}
\sum_{m=1}^{\infty}2me^{-\beta C^{(2)}_{\alpha}m^{\alpha-\frac{1}{2}+\varepsilon}}
&\le 2e^{-\beta C^{(2)}_{\alpha}} + 2\int_1^{\infty}ue^{-\beta C^{(2)}_{\alpha}u^{\alpha-\frac{1}{2}+\varepsilon}}\d u\\
&= 2e^{-\beta C^{(2)}_{\alpha}} + 2\int_1^{u_0}ue^{-\beta C^{(2)}_{\alpha} u^{\alpha-\frac{1}{2}+\varepsilon}}\d u+ 2\int_{u_0}^{\infty}ue^{-\beta C^{(2)}_{\alpha}u^{\alpha-\frac{1}{2}+\varepsilon}}\d u\\
&\le 2e^{-\beta C^{(2)}_{\alpha}}+ e^{-\beta C^{(2)}_{\alpha}}u_0^2 +\int_{u_0}^{\infty}ue^{-\beta C^{(2)}_{\alpha}\ln u}\d u\\
&\le 2e^{-\beta C^{(2)}_{\alpha}}+ e^{-\beta C^{(2)}_{\alpha}}u_0^2 +\int_{u_0}^{\infty}u^{1-\beta C^{(2)}_{\alpha}}\d u\\
&= 2e^{-\beta C^{(2)}_{\alpha}}+ e^{-\beta C^{(2)}_{\alpha}}u_0^2 +\frac{1}{\beta C^{(2)}_{\alpha}-2}u_0^{2-\beta C^{(2)}_{\alpha}}\\
&\le e^{-\beta C^{(2)}_{\alpha}+\ln 2}+ e^{-\beta C^{(2)}_{\alpha}+2\ln u_0}+e^{(2-\beta C^{(2)}_{\alpha})\ln u_0}\\
&\le 3e^{-\frac{2}{3}\beta C^{(2)}_{\alpha}}\\
&=e^{-\frac{2}{3}\beta C^{(2)}_{\alpha}+\ln 3}\\
&\le e^{-\frac{1}{2}\beta C^{(2)}_{\alpha}}
\end{align*}
\end{proof}

\begin{lemma}\label{lem:supp}
For every contour $\Gamma$ and $M\ge 1$,
$$
\max_{|\Gamma|=M} |\supp(\Gamma)|\le c|\Gamma|^3.
$$
\end{lemma}

\begin{proposition}\label{incompatible}
For every $\frac{1}{2}<\alpha<1$ and for sufficiently large $\beta$, we have, for every fixed contour $\Gamma$,
$$
\sum_{\Gamma'\nsim \Gamma}\vartheta_{N_k,n_k,\varepsilon}(\Gamma')\le \frac{3c}{J}\rho_{N_k,n_k,\varepsilon}\left(\beta,\alpha,\frac{1}{2}\right)H^f_{\Lambda_{N_k}}[\Gamma].
$$
\end{proposition}

\begin{proof}
Two contours $\Gamma$ and $\Gamma'$ are incompatible if one of the following events occurs,
\begin{enumerate}
\item[(a)] $B_0(\Gamma,\Gamma')$ is the event that there exist $T \in \Gamma$ and $T'\in \Gamma'$ such that 
$$
\d(T,T')\le c\min\{|\Gamma|,|\Gamma'|\}^3.
$$
\item[(b)] $B_{1,1}(\Gamma,\Gamma')$ is the event that there exist $T_1,T_2$ in $
\Gamma$ such that 
$$ 
\Delta(T_1)\subseteq \Delta(\Gamma') \quad \text{and}\quad \Delta(T_2)\cap \Delta(\Gamma')=\emptyset.
$$
\item[(c)] $B_{1,2}(\Gamma,\Gamma')$  is the event that there exist $T'_1,T'_2$ in $\Gamma'$ such that 
$$
\Delta(T'_1)\subseteq \Delta(\Gamma) \quad \text{and}\quad \Delta(T'_2)\cap \Delta(\Gamma)=\emptyset.
$$
\end{enumerate}
Splitting the sum,
$$ 
\sum_{\Gamma'\nsim \Gamma}\vartheta_{N_k,n_k,\varepsilon}(\Gamma') = 
\sum_{\Gamma'}\vartheta_{N_k,n_k,\varepsilon}(\Gamma')\mathbbm{1}_{B_0(\Gamma,\Gamma')}+\sum_{\Gamma'}\vartheta_{N_k,n_k,\varepsilon}(\Gamma')\mathbbm{1}_{B_{1,1}(\Gamma,\Gamma')}+
\sum_{\Gamma'}\vartheta_{N_k,n_k,\varepsilon}(\Gamma')\mathbbm{1}_{B_{1,2}(\Gamma,\Gamma')}.
$$
Let us bound each summation. By using that the activity function $\vartheta_{\Lambda}$ is translation invariant,
\begin{align*}
\sum_{\Gamma'}\vartheta_{N_k,n_k,\varepsilon}(\Gamma')\mathbbm{1}_{B_0(\Gamma,\Gamma')}
&\le \sum_{\Gamma'} \sum_{T\in \Gamma}\vartheta_{N_k,n_k,\varepsilon}(\Gamma')\mathbbm{1}_{\d(T,\Gamma')\le c\min\{|\Gamma|,|\Gamma'|\}^3}\\
&\le \sum_{\Gamma'} \sum_{x^*\in \sf^*(\Gamma)} \vartheta_{N_k,n_k,\varepsilon}(\Gamma')\mathbbm{1}_{\d(x^*,\Gamma')\le c\min\{|\Gamma|,|\Gamma'|\}^3}\\
&\le \sum_{\Gamma'} \sum_{x^*\in \sf^*(\Gamma)} \vartheta_{N_k,n_k,\varepsilon}(\Gamma')\mathbbm{1}_{\d(x^*,\Gamma')\le c |\Gamma'|^3}\\
&\le |\sf^*(\Gamma)|  \sum_{\Gamma'} \vartheta_{N_k,n_k,\varepsilon}(\Gamma')\mathbbm{1}_{\d(0,\Gamma')\le c |\Gamma'|^3}.
\end{align*}
Using 
$$
|\sf^*(\Gamma)|=\sum_{T\in \Gamma} 2 \le 4 (\text{number of spin flips in }\Gamma)\le \frac{4H^f_{\Lambda_{N_k}}[\Gamma]}{J}
$$
(note that some triangles can be boundaries) and the factor $(2|\supp(\Gamma')|+2c|\Gamma'|^3)\mathbbm{1}_{\Gamma'\ni 0}$ that is an upper bound for all of those translations such that $\d(0,\Gamma')\le c |\Gamma'|^3$, we have
$$
\sum_{\Gamma'}\vartheta_{N_k,n_k,\varepsilon}(\Gamma')\mathbbm{1}_{B_0(\Gamma,\Gamma')}
\le \frac{ 4H^f_{\Lambda_{N_k}}(\Gamma)}{J} \sum_{\Gamma'} \vartheta_{N_k,n_k,\varepsilon}(\Gamma') (2|\supp(\Gamma')|+2c|\Gamma'|^3)\mathbbm{1}_{\Gamma'\ni 0}.
$$
 By Lemma \ref{lem:supp}, we have
$$
 \sum_{\Gamma'}\vartheta_{N_k,n_k,\varepsilon}(\Gamma')\mathbbm{1}_{B_0(\Gamma,\Gamma')}
 \le \frac{ 16c H^f_{\Lambda_{N_k}}(\Gamma)}{J} \sum_{\Gamma'\ni 0} |\Gamma'|^3\vartheta_{N_k,n_k,\varepsilon}(\Gamma').
$$
 Let us bound the other terms:
 \begin{align*}
 \sum_{\Gamma'}\vartheta_{N_k,n_k,\varepsilon}(\Gamma')\mathbbm{1}_{B_{1,1}(\Gamma,\Gamma')} 
 &\le \sum_{\Gamma'}\vartheta_{N_k,n_k,\varepsilon}(\Gamma')\mathbbm{1}_{\{\Delta(T)\subset \Delta(\Gamma') \text{ for some }T\in \Gamma\}}\\
 &\le \sum_{\Gamma'}\vartheta_{N_k,n_k,\varepsilon}(\Gamma') \sum_{T\in \Gamma}\mathbbm{1}_{\{\Delta(T)\subset \Delta(\Gamma')\}}\\
 &\le \sum_{\Gamma'} \sum_{x^*\in \sf^*(\Gamma)} \vartheta_{N_k,n_k,\varepsilon}(\Gamma')\mathbbm{1}_{\{x^*\in \Delta(\Gamma')\}}\\
 &\le  |\sf^*(\Gamma)| \sum_{\Gamma'} \vartheta_{N_k,n_k,\varepsilon}(\Gamma')\mathbbm{1}_{\{0\in \Delta(\Gamma')\}}\\
 &\le \frac{4H_{\Lambda_{N_k}}^f(\Gamma)}{J}\sum_{\Gamma'} \vartheta_{N_k,n_k,\varepsilon}(\Gamma')\mathbbm{1}_{\{0\in \Delta(\Gamma')\}}.
 \end{align*}
 and also, by using that $\vartheta_{N_k,n_k,\varepsilon}$ is translation invariant,
 \begin{align*}
 \sum_{\Gamma'}\vartheta_{N_k,n_k,\varepsilon}(\Gamma')\mathbbm{1}_{B_{1,2}(\Gamma,\Gamma')}
 &=  \sum_{\Gamma'} \vartheta_{N_k,n_k,\varepsilon}(\Gamma')\sum_{T'_1\neq T'_2 \in \Gamma'}\mathbbm{1}_{\{\Delta(T'_1) \subseteq \Delta(\Gamma)\}}\mathbbm{1}_{\{\Delta(T'_2)\cap \Delta(\Gamma)=\emptyset\}}\\
 &\le \sum_{\Gamma'} \vartheta_{N_k,n_k,\varepsilon}(\Gamma')\sum_{T\in\Gamma}\sum_{T'_1\neq T'_2 \in \Gamma'}\mathbbm{1}_{\{\Delta(T'_1) \subseteq \Delta(T)\}}\mathbbm{1}_{\{\Delta(T'_2)\cap \Delta(T)=\emptyset\}}\\
 &\le \sum_{\Gamma'} \vartheta_{N_k,n_k,\varepsilon}(\Gamma')\sum_{T\in\Gamma}\sum_{T'_1\neq T'_2 \in \Gamma'} (\mathbbm{1}_{0\in \Delta(T'_1)} + \mathbbm{1}_{0\in \Delta(T'_2)})\\
 &\le  \sum_{\Gamma'} \vartheta_{N_k,n_k,\varepsilon}(\Gamma')\sum_{T\in\Gamma}\sum_{T'_1\neq T'_2 \in \Gamma'}\mathbbm{1}_{0\in \Delta(\Gamma')}\\
 &\le |\sf^*(\Gamma)| \sum_{\Gamma'} \vartheta_{N_k,n_k,\varepsilon}(\Gamma')\sum_{T'_1\neq T'_2 \in \Gamma'}\mathbbm{1}_{0\in \Delta(\Gamma')}\\
 &\le  \frac{4H^f_{\Lambda_{N_k}}(\Gamma)}{J} \sum_{\Gamma'} \vartheta_{N_k,n_k,\varepsilon}(\Gamma')  |\sf^*(\Gamma')|^2\mathbbm{1}_{0\in \Delta(\Gamma')}.
 \end{align*}
 Thus, we get
 \small 

\begin{align*}
  &\sum_{\Gamma'\nsim \Gamma}\vartheta_{N_k,n_k,\varepsilon}(\Gamma')\\
   &\le 
  \frac{4H^f_{\Lambda_{N_k}}(\Gamma)}{J} \left(   4c\sum_{\Gamma'\ni 0} |\Gamma'|^3\vartheta_{N_k,n_k,\varepsilon}(\Gamma')  +\sum_{\Gamma'} \vartheta_{N_k,n_k,\varepsilon}(\Gamma')\mathbbm{1}_{\{0\in \Delta(\Gamma')\}}+\sum_{\Gamma'} \vartheta_{N_k,n_k,\varepsilon}(\Gamma')  |\sf^*(\Gamma')|^2\mathbbm{1}_{0\in \Delta(\Gamma')} \right).
 \end{align*}

\normalsize
 Since $1\le  |\sf^*(\Gamma')|^2 \le c|\Gamma'|^3$, we conclude
$$
\sum_{\Gamma'\nsim \Gamma}\vartheta_{N_k,n_k,\varepsilon}(\Gamma') \le 
6c  \frac{4H^f_{\Lambda_{N_k}}(\Gamma)}{J}\sum_{\Gamma'\ni 0} |\Gamma'|^3\vartheta_{N_k,n_k,\varepsilon}(\Gamma') =
\frac{24cH^f_{\Lambda_{N_k}}(\Gamma)}{J}\sum_{\Gamma'\ni 0} |\Gamma'|^3\vartheta_{N_k,n_k,\varepsilon}(\Gamma').
$$
 By Remark 5.3 in \cite{LP}, using the fact that for $\alpha_0>0$, 
$$
 |\Gamma|\le 6\left(\frac{\alpha(1-\alpha)}{2}\right)^{1/\alpha_0}H^f_{\Lambda_{N_k}}(\Gamma)^{1/\alpha_0}
$$
 for all contours $\Gamma$, we have
$$
 |\Gamma|^3 \le 6^3\left(\frac{\alpha(1-\alpha)}{2}\right)^{3/\alpha_0} \exp\left(\frac{3}{\alpha_0} \ln H^f_{\Lambda_{N_k}}(\Gamma)\right).
$$
 Using that $\ln x \le x$ for all $x>0$, and reminding ourselves that $1-2K_c(\alpha)<0$, we have, for sufficiently large $\beta>0$,
\small
 \begin{align*}
 |\Gamma'|^3\vartheta_{N_k,n_k,\varepsilon}(\Gamma') 
 &\le 6^3\left(\frac{\alpha(1-\alpha)}{2}\right)^{3/\alpha_0} \exp\left( \frac{3}{\alpha_0}H^f_{\Lambda_{N_k}}(\Gamma) -\beta(2K_c(\alpha)-1)H^f_{\Lambda_{N_k}}(\Gamma) +2\beta\sum_{T\in \Gamma }\sum_{x\in T}\tilde{h}_{N_k,n_k,\varepsilon}(x)\right)\\
\normalsize
 &\le  \exp\left( \frac{-\beta}{2}(2K_c(\alpha)-1)H^f_{\Lambda_{N_k}}(\Gamma) +2\beta\sum_{T\in \Gamma }\sum_{x\in T}\tilde{h}_{N_k,n_k,\varepsilon}(x)\right)
 \end{align*}
 as we desired.
\end{proof}

\begin{theorem}
For every $\frac{1}{2}<\alpha<1$ and sufficiently large $\beta$, the series (\ref{absconvbarxi}) is absolutely convergent. Moreover,
$$
 \log \Upsilon_{N_k,n_k,\varepsilon} = \sum_{n=1}^{\infty} \vartheta_{N_k,n_k,\varepsilon}(\Gamma)(1+\mathcal{R}_{N_k,n_k,\varepsilon}(\Gamma)),
$$
where
$$ 
\mathcal{R}_{N_k,n_k,\varepsilon}(\Gamma) = \sum_{n=2}^{\infty}\frac{1}{n!}\sum_{(\Gamma_2,\ldots,\Gamma_n) \in \un{\Gamma}_{\Lambda_{N_k}}^{n-1}}\phi^T(\Gamma,\Gamma_2,\ldots,\Gamma_n) \prod_{i=2}^n \vartheta_{N_k,n_k,\varepsilon}(\Gamma_i),
$$
and the following inequality holds,
\begin{equation}\label{eq:rbound}
|\mathcal{R}_{N_k,n_k,\varepsilon}(\Gamma)| \le (e^{H^f_{\Lambda_{N_k}}[\Gamma]}-1) g(\beta,\alpha)
\end{equation}
where
$$
g(\beta,\alpha)=\frac{\frac{12c}{J}\rho_{N_k,n_k,\varepsilon}\left(\beta,\alpha,\frac{1}{4}\right)}{1-\frac{12c}{J}\rho_{N_k,n_k,\varepsilon}\left(\beta,\alpha,\frac{1}{4}\right)}
$$
and $\beta$ is chosen large enough such that $\frac{6c}{J}\rho_{N_k,n_k,\varepsilon}(\beta,\alpha,1/4)<1$.
\end{theorem}

\begin{proof}
For $n\ge 2$, define
$$ 
\mathcal{R}_{n}(\Gamma)=\sum_{(\Gamma_2,\ldots,\Gamma_n) \in \un{\Gamma}_{\Lambda_{N_k}}^{n-1}}\phi^T(\Gamma,\Gamma_2,\ldots,\Gamma_n) \prod_{i=2}^n \vartheta_{N_k,n_k,\varepsilon}(\Gamma_i),
$$
Assuming that $\mathcal{G}_{\{\Gamma_1,\ldots,\Gamma_n\}}$ is connected, we have
$$
\phi^T(\Gamma_1,\ldots,\Gamma_n) = \sum_{\substack{G\in \mathcal{G}_{\{\Gamma_1,\ldots,\Gamma_n\}}}}(-1)^{e(G)}
$$
Define $\mathcal{G}_n$ to be the set of connected graphs with $n$ vertices. Note that
$$
\sum_{(\Gamma_2,\ldots,\Gamma_n) \in \un{\Gamma}_{\Lambda_{N_k}}^{n-1}}\sum_{\substack{G\in \mathcal{G}_{\{\Gamma,\Gamma_2,\ldots,\Gamma_n\}}}} (\cdot)= \sum_{G^* \in \mathcal{G}_n} \sum_{\substack{(\Gamma_2,\ldots,\Gamma_n) \in \un{\Gamma}_{\Lambda_{N_k}}^{n-1} \\ \mathcal{G}_{\{\Gamma,\Gamma_2,\ldots,\Gamma_n\}}=G^*}} \sum_{\substack{G\in \mathcal{G}_n \\ G\subset G^*}}(\cdot)
$$
By using the Rota formula,
$$
\left| \sum_{\substack{G\in \mathcal{G}_{n} \\ G\subset G^*}}(-1)^{e(G)} \right|
\le N(G^*)=\sum_{\substack{\tau \in T_n \\ \tau \subset G^*}}1,
$$
where $N(G^*)$ is the number of spanning tree graphs in $G^*$ and $T_n$ is the set of trees with $n$ vertices, the bound gives us
$$
|\mathcal{R}_{n}(\Gamma)|\le 
\sum_{G^* \in \mathcal{G}_n} \sum_{\substack{\tau \in T_n \\ \tau \subset G^*}}\sum_{\substack{(\Gamma_2,\ldots,\Gamma_n) \in \un{\Gamma}_{\Lambda_{N_k}}^{n-1} \\ \mathcal{G}_{\{\Gamma,\Gamma_2,\ldots,\Gamma_n\}} \supset \tau}} \prod_{i=2}^n \vartheta_{N_k,n_k,\varepsilon}(\Gamma_i).
$$
Define for any subtree $\tau$, 
$$
D_{\Gamma}(\tau)=\sum_{\substack{(\Gamma_2,\ldots,\Gamma_n) \in \un{\Gamma}_{\Lambda_{N_k}}^{n-1} \\ \mathcal{G}_{\{\Gamma,\Gamma_2,\ldots,\Gamma_n\}} \supset \tau}} \prod_{i=2}^n \vartheta_{N_k,n_k,\varepsilon}(\Gamma_i).
$$
Let us bound $D_{\Gamma}(\tau)$. Fix the root on the vertex associated with the contour $\Gamma$.  Let $V_1$ be the vertices in $\tau$ with degree one, that is, the  leaves.  Define $\Gamma_{V_1}$ to be the contours except $\Gamma$ connected to some vertex of $V_1$. For each $\Gamma_{i,1}$ in $\Gamma_{V_1}$, we have, by Proposition \ref{incompatible}, that
\be\label{eq:boundnsim}
\sum_{\Gamma'\nsim \Gamma_{i,1}}\vartheta_{N_k,n_k,\varepsilon}(\Gamma')\le \frac{3c}{J}\rho_{N_k,n_k,\varepsilon}\left(\beta,\alpha,\frac{1}{2}\right)H^f_{\Lambda_{N_k}}[\Gamma_{i,1}].
\ee
Note that if $d$ is the incidence number of $\Gamma_{i,1}$, then it is associated with $d-1$ leaves, so it receives a contribution
$$
\left( \frac{3c}{J}\rho_{N_k,n_k,\varepsilon}\left(\beta,\alpha,\frac{1}{2}\right)H^f_{\Lambda_{N_k}}[\Gamma_{i,1}] \right)^{d-1}.
$$
To iterate the procedure, note that $u^d/n! \le e^u$ for all $u\ge 0$ and $d\ge 0$. By (\ref{eq:boundnsim}), we have
\begin{align*}
&\sum_{\Gamma'\nsim \Gamma_{i,1}}(H^f_{\Lambda_{N_k}}[\Gamma'])^{d-1}\vartheta_{N_k,n_k,\varepsilon}(\Gamma')
\le (d-1)!\sum_{\Gamma'\nsim \Gamma_{i,1}} e^{H^f_{\Lambda_{N_k}}[\Gamma']}\vartheta_{N_k,n_k,\varepsilon}(\Gamma')\\
&\le (d-1)!\sum_{\Gamma'\nsim \Gamma_{i,1}}
\exp\left(-(\beta(2K_c(\alpha)-1)-1)H^f_{\Lambda_{N_k}}[\Gamma] +2\beta \sum_{T\in \Gamma_0}\sum_{x\in T}\tilde{h}^{\eta}_{N_k,n_k,\varepsilon}(x)\right)\\
&\le (d-1)!\sum_{\Gamma'\nsim \Gamma_{i,1}}
\exp\left(-\beta\frac{(2K_c(\alpha)-1)}{2}H^f_{\Lambda_{N_k}}[\Gamma] +2\beta \sum_{T\in \Gamma_0}\sum_{x\in T}\tilde{h}^{\eta}_{N_k,n_k,\varepsilon}(x)\right)\\
&\le (d-1)!\frac{3c}{J}\rho_{N_k,n_k,\varepsilon}\left(\beta,\alpha,\frac{1}{4}\right)H^f_{\Lambda_{N_k}}[\Gamma_{i,1}].
\end{align*}
For the root $\Gamma$, we should replace $d-1$ by $d$. Calling $C^{(3)}(\beta,\alpha)=\frac{3c}{J}\rho\left(\beta,\alpha,\frac{1}{4}\right)$, we have, for a fixed tree $\tau$,
\begin{align*}
D_{\Gamma}(\tau) &\le \{C^{(3)}(\beta,\alpha)H^f_{\Lambda_{N_k}}[\Gamma]\}^{d_1}\prod_{i=2}^n (d_i-1)!C^{(3)}(\beta,\alpha)^{d_i-1}\\
& = C^{(3)}(\beta,\alpha)^{n-1}H^f_{\Lambda_{N_k}}[\Gamma]^{d_1}\prod_{i=2}^n(d_i-1)!.
\end{align*}
Since there are exactly $\binom{n-2}{d_1,d_2-1,\ldots, d_n-1}$ of those trees, it holds
\begin{align*}
\sum_{G^* \in \mathcal{G}_n} \sum_{\substack{\tau\in T_n \\ \tau \subset G^*}}D_{\Gamma}(\tau) 
&\le C^{(3)}(\beta,\alpha)^{n-1} \sum_{\substack{d_1+\ldots+d_n=2n-2 \\ d_i\ge 1}}H^f_{\Lambda_{N_k}}[\Gamma]^{d_1}\binom{n-2}{d_1,d_2-1,\ldots, d_n-1}\prod_{i=2}^n(d_i-1)!\\
&\le (n-2)! C^{(3)}(\beta,\alpha)^{n-1} \sum_{d_1=1}^{\infty}\frac{H^f_{\Lambda_{N_k}}[\Gamma]^{d_1}}{d_1!}\sum_{\substack{d_2+\ldots+d_n=2n-2-d_1 \\ d_i\ge 1}}1\\
&\le (n-2)! C^{(3)}(\beta,\alpha)^{n-1} \sum_{d_1=1}^{\infty}\frac{H^f_{\Lambda_{N_k}}[\Gamma]^{d_1}}{d_1!}\binom{2n-3}{n-2}\\
&\le (n-2)! (4C^{(3)}(\beta,\alpha))^{n-1} \sum_{d_1=1}^{\infty}\frac{H^f_{\Lambda_{N_k}}[\Gamma]^{d_1}}{d_1!}\\
&\le (n-2)! (4C^{(3)}(\beta,\alpha))^{n-1} (e^{H^f_{\Lambda_{N_k}}[\Gamma]}-1)
\end{align*}
Thus,
\begin{align*}
\mathcal{R}_{N_k,n_k,\varepsilon}(\Gamma) &= \sum_{n=2}^{\infty}\frac{\mathcal{R}_n(\Gamma)}{n!}\\
&\le \sum_{n=2}^{\infty}\frac{1}{n!}(n-2)! (4C^{(3)}(\beta,\alpha))^{n-1} (e^{H^f_{\Lambda_{N_k}}[\Gamma]}-1)\\
&\le (e^{H^f_{\Lambda_{N_k}}[\Gamma]}-1)\sum_{n=2}^{\infty}\frac{(4C^{(3)}(\beta,\alpha))^{n-1}}{n(n-1)}.
\end{align*}
Taking $\beta$ large enough such that $C^{(3)}(\beta,\alpha)<1/4$, we conclude that the series absolutely converges.
\end{proof}

\begin{proposition}\label{propclusterexp}
For every $\frac{1}{2}<\alpha<1$, if $\beta>0$ is such that $g(\beta,\alpha)<1/2$, the series
\be \label{Exppluseta}
\log \tilde{\Xi}^{\eta}_{N_k,n_k} = \sum_{n=1}^{\infty}\frac{1}{n!}\sum_{R_1}\ldots \sum_{R_n}\phi^T(R_1,\ldots,R_n)\prod_{i=1}^n \xi^{\eta}_{N_k,n_k}(R_i)
\ee
is absolutely convergent.
\end{proposition}

Note that
$$ 
\log  \tilde{\Xi}^{\eta}_{N_k,n_k}-\log  \tilde{\Xi}^{-\eta}_{N_k,n_k} =\int_{-1}^1 \frac{\partial}{\partial t}\log   \tilde{\Xi}^{\eta}_{N_k,n_k}(t) \d t,
$$
where $\tilde{\Xi}^{\eta}_{N_k,n_k}(t)$ is defined by
$$
 \tilde{\Xi}^{\eta}_{N_k,n_k}(t) = \sum_{\sigma \in \Omega^+_{\Lambda}}e^{-\beta H^f_{\Lambda_{N_k}}(\sigma) - 2\beta t \tilde{E}^{\eta}_{N_k,n_k}(\sigma)}.
$$
Then, using that 
$$
|\tilde{E}^{\eta}_{N_k,n_k}(\sigma)| = \left|\sum_{x\in \Lambda_{N_k}}\tilde{h}^{\eta}_{N_k,n_k}(x)\mathbbm{1}_{\sigma_x=-1} \right| \le \sum_{x\in \Lambda_{N_k}}|\tilde{h}^{\eta}_{N_k,n_k}(x)|
$$
 for all $\sigma\in \Omega^+_{\Lambda}$,
$$
|\log  \tilde{\Xi}^{\eta}_{N_k,n_k}-\log  \tilde{\Xi}^{-\eta}_{N_k,n_k} | \le 2\beta  \sum_{x\in \Lambda_{N_k}}|\tilde{h}^{\eta}_{N_k,n_k}(x)|.
$$
Therefore, conditioned to the good boundary conditions $\eta\in \Omega_0$,
\be \label{good-field}
\left|\frac{1}{2\beta}(\log  \tilde{\Xi}^{\eta}_{N_k,n_k}-\log  \tilde{\Xi}^{-\eta}_{N_k,n_k}) \right| \le  \sum_{x\in \Lambda_{N_k}}|\tilde{h}^{\eta}_{N_k,n_k}(x)|\le \frac{2(3^{\alpha-\frac{1}{2}+\varepsilon})}{\alpha-\frac{1}{2}+\varepsilon} (N_k-n_k)^{\alpha-\frac{1}{2}+\varepsilon}.
\ee

This estimate (\ref{good-field}) for good boundary conditions $\eta$'s and  decoupled enough subsequences $(n_k,N_k)_k$ enables us to show  that,  
depending on the asymptotic behavior of the extra stochastic term $W^\eta$, we  get either a well-defined distribution on the asymptotic random weights in the extremal decomposition of the mixed states in the support of the metastate or degeneracy of this limiting distribution which leads to a concentration of the metastate  on the extremal maximal and minimal states $\mu^+$ and $\mu^-$ (similar to higher-dimensional systems, occurring for slow decays). 

To conclude our proofs, it remains to prove Propositions 1 and 5 to use the WLLT in the intermediate case (in next Section) and to put afterwards everything together to get Theorem 2 and 3 (in the Appendix Section 9).

\section{Proof of Propositions \ref{mainprop} and \ref{prop_convergence_nu}}\label{sec:proofs}

We detail now the proof of  Propositions \ref{mainprop}, and describe the one of Proposition \ref{prop_convergence_nu} afterwards.

Let us choose $k, N_k, n_k, \beta,\epsilon$ such that the measure $\mu$ is after decoupling ($n_k <N_k^{\frac{1}{2} ( \alpha- \frac{1}{2})}$), at a temperature low enough such that a cluster expansion converges for good $\eta's$ from Proposition 6, and call $\mathcal{G}$ the $\sigma$-algebra generated by such disorder variables $\eta \in \Omega_0$. In particular, we shall use that for such good $\eta$'s, and for $N=N_k-n_k$,
$$
\frac{1}{2 \beta} \Big| \ln \Xi_N^{\eta} - \ln \Xi_N^{-,\eta}\Big| \leq C_0 \cdot N^{\alpha-\frac{1}{2}+\varepsilon}
$$

As seen in Section 5, it  is equivalent (to proving) that for all $\epsilon >0$ and $\alpha > \frac{1}{2}$,

$$
\limsup_N N^{\alpha- \frac{1}{2} -\epsilon}  \mathbb{P} \Big[ \mid F_N^\eta \mid \leq  K_N(\beta, \tau)\Big] <  \infty
$$

with $K_N(\beta, \tau) = \frac{1}{2\beta}  \log \left( \frac{\lVert \nu^{\eta}_{\Lambda_N}-\nu^{-\eta}_{\Lambda_N} \rVert_X}{\tau-\ell}-1\right)$ and $F_N^\eta= \ln \mathcal{Z}_N^{+,\eta}$ can be expressed at low $T$ (high $\beta_1$)  as
$$
F_N^\eta = W_N^\eta + \frac{1}{2 \beta} \Big( \ln \Xi_N^{\eta} - \ln \Xi_N^{-,\eta}\Big)
$$

To do so, we shall use our criteria for the WLLT at $T=T_1=\frac{1}{\beta_1}>0$ and  find $A_N^T, \delta_N^T, \tau_N^T$ satisfying the conditions of Proposition \ref{LLT} with the correct speed (and $\beta$ adjusted so that $K_N(\beta,\tau)$ is of the order of $\delta_N^T$).

 Denote 
by $\Psi_N^T$ the characteristic function of $F_N$ at (low) temperature $T$, for $t$ in an admissible interval. Thus, to extend the inner-temperature-$T_1=0$ toy-models case to low nonzero temperatures via cluster expansions, we get quantities such that 

\be \label{ANT}
\frac{A_N^T}{\delta_N^T}=C_1 N^{\alpha- \frac{1}{2} -\epsilon }
\ee

and
\begin{enumerate}\label{ToolWLLT}
    \item[(i)] $\limsup_{N\to \infty}A_N^T\int_{-\tau_N^T}^{\tau_N^T}|\psi_N^T(t)|dt \le 2\pi$
    \item[(ii)] There exists $k>1$ such that $\lim_{n\to \infty}\frac{A_n^T}{\delta_n^k \tau_n^{k-1}}=0$
\end{enumerate}

At $T=0$, denote $\Psi_N^0, \tau_N^0, A_N^0$ the corresponding quantities used for the toy model in Section 4, so that

$$
\Psi_N^0(t)= \mathbb{E}\big[e^{i t W_N^\eta }\big],\; A_N^0 = C_2 \cdot N^{\alpha-\frac{1}{2}}, \tau_N^0 = C_3 \cdot N^{1-\alpha}
$$
and $(i), (ii)$ below are satisfied.

To get the correct $A_N^T$ and $\delta_N^T$ so that $(i)$ holds, we can proceed as in \cite{ENS1} by a Mayer expansion of the characteristic function at low temperature, where our corridor $N=N_k-n_k$ plays the role of the boundary $\partial \Lambda \setminus \partial \Lambda_C$ in (87) of \cite{ENS1}, page 1032. Write, for $t$ in some admissible interval  $[-\tau_N^T, + \tau_N^T]$,

$$
\Psi_N^T(t) = \mathbb{E}\big[e^{i t F_N^\eta} \big]= \mathbb{E}\big[e^{i t W_N^\eta + \frac{it}{2 \beta} \big( \ln \Xi_N^{\eta} - \ln \Xi_N^{-,\eta} \big) } \big]
$$

and condition first on (typically good) $\eta's$,

$$
\Psi_N^T(t) =  \mathbb{E}\big[\mathbb{E} \big[e^{i t W_N^\eta + \frac{it}{2 \beta} \big( \ln \Xi_N^{\eta} - \ln \Xi_N^{-,\eta} \big) } | \mathcal{G}\big] (\eta)  \big]
$$




First, consider adapted quantities of the form
$
\tau_N^T = \tau_N^0 \cdot N^{-\delta'}
$
, with $\delta'>0$ such that (\ref{ToolWLLT}) holds, and
$
A_N^T = A_N^0 \cdot I_\beta(N)
$
where
\begin{equation}\label{IBetaN}
I_\beta(N) = \frac{\int_{-\tau_N^T}^{\tau_N^T} \mid \Psi_N^T(t) \mid dt}{{\int_{-\tau_N^0}^{\tau_N^0} \mid \Psi_N^0(t) \mid dt}{} }
\end{equation}

Then one proves the result using the low-temperature Mayer expansion of the characteristic function of the free energy similar to the higher dimensional cases of \cite{ENS1} -- admissible thanks to our cluster expansions valid at low temperature from Section 6 --, and we end up (from there) with the result that we can take $I(N) = N^{-\delta}$ for any $\delta < \alpha - \frac{1}{2}$. 

Therefore, at low enough temperature and for $\alpha > \frac{1}{2}$, for all $\epsilon >0$ we can choose $\delta < \alpha - \frac{1}{2}$ and $\delta'$ such that conditions $(i),(ii)$ hold for the WLLT at low temperature with
$$
A_N^T = A_N^0 \cdot N^{-\delta}, \; \tau_N^T = \tau_N^0 \cdot N^{-\delta'}, \; \delta_N^T = N^{\epsilon} \delta_N^0
$$   
leading the an analog control of the characteristic functions as in (95) of \cite{ENS1} p 1034, with our corridor $N$, and for good $\eta$'s :
$$
\mid \Psi_N^\eta(t) \mid \leq e^{\frac{1}{2} \beta^2 t^2 N}.
$$

From this, we get the WLLT at low $T_1$ with
$$
\limsup_N \frac{A_N^T}{\delta_N^T} \mathbb{P} [a \delta_N^T \leq F_N^\eta \leq b \delta_N^T] < b-a
$$

or with the choices above with $\epsilon < \alpha - \frac{1}{2}$, for our subsequences $N$,

$$
\mathbb{P} [a N^{\epsilon} \leq F_N^\eta \leq b N^{\epsilon} ] = \mathcal{O} \Big(N^{\frac{1}{2} - \alpha - \epsilon}\Big)
$$
which leads to Proposition \ref{mainprop}. We  derive the Mayer expansions in Appendix 4 and describe more precisely the use of free boundary condtions and boundary contours.



To get Proposition \ref{prop_convergence_nu} requires proving that the constrained  measure $\nu_N^\eta$ converges to the extremal state $\mu^+$ and is based on the idea that the characteristic function  of the free energy difference behaves similarly to the case of  the toy model of Section \ref{SectionToyModel}. 

Denote again $F_N$ for our free energy $F=\ln Z$. As approximately $$\mu_N = \frac{1}{(e^{\beta F_N} + e^{- \beta F_N)}}( (e^{\beta F_N}) \mu_N^{+} +( e^{- \beta F_N}))  \mu_N^{-})$$  the weights will only be non-trivial if the sequence $|F_N|$ does not diverge. Its behaviour follows now from the following scheme, already described in various places of this paper :

\begin{enumerate}
	\item The free energy (and thus its associated characteristic function)  can be developed into a low-temperature cluster expansion with the (ground state) energy being the leading term. The characteristic function thus becomes a product of terms for individual sites, times well-controled weights containing the interactions.
	\item This holds for "good" boundary conditions $\eta$, and then the contour energies are analytic in $t$ in some interval, as follows from the arguments  in  \cite {BEEKR}, 
	\item Moreover, the probability of $\eta$ being bad is small enough. Here 
	"Good" means that$ |h_i|< |i|^{\frac{1}{2} - \alpha + \varepsilon}$ for all $i$ larger than $n$ and $n = N^{\delta}$, and one uses it to extract good subsequences using again a Borel-Cantelli framework.
	\item In addition, the contribution of the first $n$ sites can be neglected in comparison to the total $N+n$, which is shown by the rigorous derivation of cluster expansions with or without decoupling from Section 6.
\end{enumerate}
 This led to Proposition \ref{prop_convergence_nu} and can be formally performed in the same way as Proposition 1, but with the additional constraint of having a mixed $(+,\eta)$-b.c. and not only $\eta$'s. Proposition 1 provides a speed of convergence, but with the additional constraint here we only need to get convergence, which is easier. 
\section{Conclusion}
In this paper we have analysed the metastate behaviour of a class of simple disordered spin systems,  one-dimensional long-range ferromagnets with polynomially decaying interaction.  Equivalently we can consider a long-range  model with Mattis disorder at fixed boundary conditions. We chose the decay slow enough that a phase transition occurs, and found that there are two different possibilities depending on the decay rate. 

 When the decay is slow enough, the metastate is dispersed and  has its support on  the plus and the minus state. This behaviour, as well as its proof, is, up to some adaptations and modifications,  analogous to what happens for nearest-neighbor models in higher dimensions, as was analysed in \cite{ENS1}.

When the decay rate is faster, the boundary terms remain bounded by an  almost surely finite random variable, and in that case the somewhat unusual phenomenon occurs that the metastate is dispersed and supported by mixed Gibbs states, rather than by  extremal ones. 

Our estimates are closely related to those of \cite{JOP} and the fact that there is a change at the value $\alpha=1/2$ (which corresponds to their $\alpha=3/2$) has  essentially the same origin.  This also appears in \cite{PreprintEFMV}.
These papers consider  one-sided (half-line) long-range extensions of dynamical systems results, answering  deep questions involving Ruelle operators, but in contrast to our work these results apply at high temperatures only.    

The same threshold,  but from a somewhat different origin, also shows up  in the Aizenman-Wehr dimensional  reduction for long-range RFIM (\cite{COP,COP1}), and also in other places involving long-range Ising models, see  \cite{Chat23, PreprintEFMV, Read}. 

\section{Appendices}\label{Appendix}
\section*{Appendix 1: Proof of Theorem \ref{thma} and \ref{thmnull}}\label{app1}

\begin{proof}[Proof of Theorem \ref{thma}]
For a fixed $1/2<\alpha<1$ and $a>0$, choose 
$$
\varepsilon< \frac{a\lf \alpha-\frac{1}{2} \ri}{a+\lf \alpha-\frac{1}{2} \ri^{-1}}
$$
and define $\beta_1(a)=\beta_0(\varepsilon)$. Let $\beta>\beta_1(a)$ and $N_k>k^{(\alpha-\frac{1}{2})^{-1}+a}$.

For a fixed $\mu\notin \{\mu^+,\mu^-\}$, there exists an open set $B\subset \mathcal{M}_1(\Omega)$ such that $\mu \in B$ but $\mu^+,\mu^-\notin \bar{B}$. Thus, there exist $X\Subset \mathbb{Z}$ and $\tau>0$ such that $B^{\tau}_X(\mu^+) \cap B=\emptyset$ and $B^{\tau}_X(\mu^-) \cap B= \emptyset$. By Proposition \ref{mainprop},
\begin{align*}
\mathbb{P}(\mu^{\eta}_{\Lambda_{N_k}}\in B)&\le
\mathbb{P}(\mu^{\eta}_{\Lambda_{N_k}} \notin B^{\tau}_X(\mu^+)\cup B^{\tau}_X(\mu^-))\\
&\le C k^{\left((\alpha-\frac{1}{2})^{-1}+a\right) \left( \frac{1}{2}-\alpha+\varepsilon \right)}\\
&= C k^{-1-a\left( \frac{1}{2}-\alpha - \varepsilon \right)+\varepsilon(\alpha-\frac{1}{2})^{-1}}.
\end{align*}
Thus
$$
\sum_{k=1}^{\infty}\mathbb{P}\left(\mu^{\eta}_{\Lambda_{N_k}}\in B\right)<\infty.
$$
By Borel-Cantelli, the set $B$ does not contain any limit points of the sequence $(\mu^{\eta}_{\Lambda_{N_k}})_{k\ge 1}$ $\mathbb{P}$-almost surely. Therefore, $\mu$ is not a limit point $\mathbb{P}$-almost surely.

Let us show that $\mu^+$ and $\mu^-$ are limit points  $\mathbb{P}$-almost surely. Note that the events $\{\mu^{\eta}_{\Lambda_{N_k}}\in B^{\tau}_X(\mu^+)\}$ are not independent for all $k\ge 1$. To solve the problem, define
$$ 
\tilde{W}^{\eta}_{N_k,N_{k+1}}(\sigma) = \sum_{N_k<|y|\le N_{k+1}}\sum_{|x|\le N_k}J_{xy}\sigma_x\eta_y
$$
and the probability measure
$$
\tilde{\mu}^{\eta}_{N_k,N_{k+1}}(\sigma)=\frac{1}{\tilde{Z}^{\eta}_{N_k,N_{k+1}}}e^{-\beta H^{\free}_{\Lambda_{N_k}}(\sigma)+\beta \tilde{W}^{\eta}_{N_k,N_{k+1}}(\sigma)}
$$
where $\tilde{Z}^{\eta}_{N_k,N_{k+1}}$ is the partition function given by
$$
\tilde{Z}^{\eta}_{N_k,N_{k+1}} = \sum_{\sigma \in \Omega_{\Lambda_{N_k}}}e^{-\beta H^{\free}_{\Lambda_{N_k}}(\sigma)+\beta \tilde{W}^{\eta}_{N_k,N_{k+1}}(\sigma)}.
$$
Note that
$$  
W^{\eta}_{\Lambda_{N_k}}(\sigma) = \sum_{|y|>N_{k+1}}\sum_{|x|\le N_k}J_{xy}\sigma_x\eta_y +\tilde{W}^{\eta}_{N_k,N_{k+1}}(\sigma),
$$
and
\begin{align}
\left| \sum_{|y|>N_{k+1}}\sum_{|x|\le N_k}J_{xy}\sigma_x\eta_y \right|
&\le  \sum_{|y|>N_{k+1}}\sum_{|x|\le N_k}J_{xy} \nonumber\\
&= \sum_{y>N_{k+1}}\sum_{|x|\le N_k} \frac{1}{(y-x)^{2-\alpha}} + \sum_{y<-N_{k+1}}\sum_{|x|\le N_k}\frac{1}{(x-y)^{2-\alpha}} \nonumber\\
&= \sum_{y>N_{k+1}}\sum_{|x|\le N_k} \frac{1}{(y-x)^{2-\alpha}} + \sum_{y>N_{k+1}}\sum_{|x|\le N_k}\frac{1}{(x+y)^{2-\alpha}} \nonumber\\
&\le \sum_{y>N_{k+1}} \frac{2N_k+1}{(y-N_k)^{2-\alpha}} + \sum_{y>N_{k+1}}\frac{2N_k+1}{(y-N_k)^{2-\alpha}} \nonumber\\
&=2(2N_k+1)\sum_{y>N_{k+1}} \frac{1}{(y-N_k)^{2-\alpha}} \nonumber\\
&\le 6N_k \int_{N_{k+1}}^{\infty}\frac{1}{(u-N_k)^{2-\alpha}} du\nonumber\\
&\le \frac{6}{1-\alpha} \frac{N_k}{(N_{k+1}-N_k)^{1-\alpha}} \nonumber.
\end{align}
The bound above implies
\begin{align}
\left|W^{\eta}_{\Lambda_{N_k}}(\sigma) \right| &\le  \frac{6}{1-\alpha} \frac{N_k}{(N_{k+1}-N_k)^{1-\alpha}} + \tilde{W}^{\eta}_{N_k,N_{k+1}}(\sigma) \label{ineqtildew1}\\
\tilde{Z}^{\eta}_{N_k,N_{k+1}}  &\ge \exp\left(  -\frac{6\beta}{1-\alpha} \frac{N_k}{(N_{k+1}-N_k)^{1-\alpha}} \right)Z^{\eta}_{\Lambda_{N_k}} \label{ineqtildew2}
\end{align}

Let us control $\lVert \tilde{\mu}^{\eta}_{N_k,N_{k+1}} - \mu^+\rVert_X$ with the following bound coming from the  triangle inequality,
$$ 
\lVert \tilde{\mu}^{\eta}_{N_k,N_{k+1}} - \mu^+\rVert_X \le 
\sup_{\substack{\lVert f\rVert =1 \\ D_f\subset X}} \left|\tilde{\mu}^{\eta}_{N_k,N_{k+1}}(f)- \mu^\eta_{\Lambda_{N_k}}(f)  \right|+
\lVert \mu^{\eta}_{\Lambda_{N_k}} - \mu^+\rVert_X.
$$
Using (\ref{ineqtildew1}) and (\ref{ineqtildew2}), for any $f:\Omega \to \mathbb{R}$ such that $\lVert f\rVert=1$ and $D_f\subset X$,
\begin{align*}
&\left|\tilde{\mu}^{\eta}_{N_k,N_{k+1}}(f)- \mu^\eta_{\Lambda_{N_k}}(f)  \right|
=
\left| \sum_{\sigma \in \Omega_{\Lambda_{N_k}}} f(\sigma)e^{-\beta H^{\free}_{\Lambda_{N_k}}(\sigma)} \left( \frac{e^{\beta \tilde{W}^{\eta}_{N_k,N_{k+1}}(\sigma)}}{\tilde{Z}^{\eta}_{N_k,N_{k+1}}}
- \frac{e^{\beta W^{\eta}_{\Lambda_{N_k}}(\sigma)}}{Z^{\eta}_{\Lambda_{N_k}}} \right) \right|\\
&\le  \sum_{\sigma \in \Omega_{\Lambda_{N_k}}} e^{-\beta H^{\free}_{\Lambda_{N_k}}(\sigma)} \left| \frac{e^{\beta \tilde{W}^{\eta}_{N_k,N_{k+1}}(\sigma)}}{\tilde{Z}^{\eta}_{N_k,N_{k+1}}}
- \frac{e^{\beta W^{\eta}_{\Lambda_{N_k}}(\sigma)}}{Z^{\eta}_{\Lambda_{N_k}}} \right|\\
&\le 
 \sum_{\sigma \in \Omega_{\Lambda_{N_k}}} e^{-\beta H^{\free}_{\Lambda_{N_k}}(\sigma)+\beta \tilde{W}^{\eta}_{N_k,N_{k+1}}(\sigma)} \left| \frac{1}{\tilde{Z}^{\eta}_{N_k,N_{k+1}}}
- \frac{\exp\left(  \frac{12\beta}{1-\alpha} \frac{N_k}{(N_{k+1}-N_k)^{1-\alpha}} \right)}{\tilde{Z}^{\eta}_{N_k,N_{k+1}}} \right|\\
&= \exp\left(  \frac{12\beta}{1-\alpha} \frac{N_k}{(N_{k+1}-N_k)^{1-\alpha}} \right)-1.
\end{align*}
Choose $N_{k+1}=N_{k+1}(\alpha,\beta)>k^{(\alpha-\frac{1}{2})^{-1}+a}$ such that
$$ 
\exp\left(  \frac{12\beta}{1-\alpha} \frac{N_k}{(N_{k+1}-N_k)^{1-\alpha}} \right)-1<\frac{\tau}{2}.
$$
Thus,
\be\label{inclusiontilde}
\left\{\lVert \tilde{\mu}^{\eta}_{N_k,N_{k+1}} - \mu^+\rVert_X\ge \tau\right\} \subseteq \left\{\lVert \mu^{\eta}_{\Lambda_{N_k}} - \mu^+\rVert_X \ge \frac{\tau}{2}\right\}
\ee
and the same inclusion is true for the minus boundary condition. Thus,
{\small $$
\mathbb{P}\left( \lVert\tilde{\mu}^{\eta}_{N_k,N_{k+1}}-\mu^+ \rVert_X \land \lVert \tilde{\mu}^{\eta}_{N_k,N_{k+1}}-\mu^- \rVert_X \ge \tau\right) 
\le
\mathbb{P}\left( \lVert \mu^{\eta}_{\Lambda_{N_k}}-\mu^+ \rVert_X \land \lVert \mu^{\eta}_{\Lambda_{N_k}}-\mu^- \rVert_X \ge \frac{\tau}{2}\right).
$$}
By Proposition \ref{mainprop},
\be\label{ineq:tildemu}
\limsup_{N\to \infty}N^{\alpha-\frac{1}{2}-\varepsilon}\mathbb{P}\left( \lVert\tilde{\mu}^{\eta}_{N_k,N_{k+1}}-\mu^+ \rVert_X \land \lVert \lVert\tilde{\mu}^{\eta}_{N_k,N_{k+1}}-\mu^- \rVert_X \ge \tau\right)<\infty.
\ee

For every $X\Subset \mathbb{Z}$ and $\tau>0$ such that $B^{\tau}_X(\mu^+)\cap B^{\tau}_X(\mu^-)=\emptyset$, by the symmetry of the distribution,
$$
\mathbb{P}\left(\tilde{\mu}^{\eta}_{N_k,N_{k+1}}\in B^{\tau}_X(\mu^+)\right)=\mathbb{P}\left(\tilde{\mu}^{\eta}_{N_k,N_{k+1}}\in B^{\tau}_X(\mu^-)\right).
$$
by (\ref{ineq:tildemu}),
$$
\lim_{k\to \infty}\mathbb{P}\left(\tilde{\mu}^{\eta}_{N_k,N_{k+1}}\in B^{\tau}_X(\mu^+)\right)=\lim_{k\to \infty}\mathbb{P}\left(\tilde{\mu}^{\eta}_{N_k,N_{k+1}}\in B^{\tau}_X(\mu^-)\right)=\frac{1}{2}.
$$
Since the events $\left\{ \tilde{\mu}^{\eta}_{N_k,N_{k+1}}\in B^{\tau}_X(\mu^\pm) \right\}$ are independent for $k\ge 1$, by the  Borel–Cantelli Theorem,
$$ 
\mathbb{P}\left( \tilde{\mu}^{\eta}_{N_k,N_{k+1}}\in B^{\tau}_X(\mu^\pm) \ \text{ i.o.} \right)=1.
$$
Using
\begin{align*}
\lVert \mu^{\eta}_{\Lambda_{N_k}} - \mu^+\rVert_X 
&\le 
\sup_{\substack{\lVert f\rVert =1 \\ D_f\subset X}} \left|\tilde{\mu}^{\eta}_{N_k,N_{k+1}}(f)- \mu^\eta_{\Lambda_{N_k}}(f)  \right|+
\lVert \tilde{\mu}^{\eta}_{N_k,N_{k+1}} - \mu^+\rVert_X\\
&\le \frac{\tau}{2}+\lVert \tilde{\mu}^{\eta}_{N_k,N_{k+1}} - \mu^+\rVert_X,
\end{align*}
 we conclude
$$ 
\mathbb{P}\left( \mu^{\eta}_{\Lambda_{N_k}}\in B^{3\tau/2}_X(\mu^\pm) \ \text{ i.o.} \right)=1.
$$
By compactness arguments, the weak closure $B^{3\tau/2}_X(\mu^{\pm})$ contains a limit point $\mathbb{P}$-almost surely. The proof is finished noticing that $\{\mu^+,\mu^-\}=\cap_{X,\tau} \bar{B}^{3\tau/2}_X(\mu^{\pm})$.
\end{proof}

\begin{proof}[Proof of Theorem \ref{thmnull}]
Assume that $B$ does not contain either $\mu^+$ or $\mu^-$. To show
$$
\mathbb{P}\left( \displaystyle\lim_{k\to \infty}\frac{1}{N_k}\displaystyle\sum_{n=1}^{N_k}\mathbbm{1}_{\mu^{\eta}_{\Lambda_{n}} \in B} =0 \right)=1,
$$
it is enough to show that, for every $\delta'>0$,
$$
\mathbb{P}\left( \frac{1}{N_k}\displaystyle\sum_{n=1}^{N_k} \mathbbm{1}_{\mu^{\eta}_{\Lambda_{n}} \in B} >\delta' \ \text{i.o.} \right) =0.
$$
Note that, by Markov's inequality and by Proposition \ref{mainprop}, there exists $n_0\ge 1$  such that, for every  $n\ge n_0$, we have $\mathbb{P}\left(\mu^{\eta}_{\Lambda_n}\in B \right)\le Dn^{-\alpha+\frac{1}{2}-\varepsilon}$ for a fixed $D>0$. Then,
{\small \begin{align*}
\mathbb{P}\left( \frac{1}{N_k}\displaystyle\sum_{n=1}^{N_k} \mathbbm{1}_{\mu^{\eta}_{\Lambda_{n}} \in B} >\delta' \right) \le  \frac{1}{\delta' N_k}\displaystyle\sum_{n=1}^{N_k} \mathbb{P}\left( \mu^{\eta}_{\Lambda_{n}} \in B\right)
< \frac{n_0}{\delta' N_k}+
D\frac{N_k^{-\alpha+\frac{3}{2}-\varepsilon}}{\delta' N_k} 
=\frac{n_0}{\delta' N_k}+ \frac{D}{\delta'} N_{k}^{-\alpha+\frac{1}{2}-\varepsilon}.
\end{align*}}
Thus, using $N_k$ in the same way as in the proof of Theorem \ref{thma},
$$
\sum_{k=1}^{\infty}N_{k}^{-\alpha+\frac{1}{2}-\varepsilon}<\infty.
$$
Therefore,
$$
\displaystyle\sum_{k=1}^{\infty}\mathbb{P}\left( \displaystyle\frac{1}{N_k}\displaystyle\sum_{n=1}^{N_k} \mathbbm{1}_{\mu^{\eta}_{\Lambda_{n}} \in B} >\delta' \right) <\infty,
$$
and the proof finishes by Borel-Cantelli.

Assume $\mu^+,\mu^-\in B$. Choose $\tau>0$ and $X\Subset \mathbb{Z}$ such that  $B^{\tau}_X(\mu^+)\cup B^{\tau}_X(\mu^-)\subseteq B$ and $B^{\tau}_X(\mu^+)\cap B^{\tau}_X(\mu^-)=\emptyset$.
We want to show
$$ 
\mathbb{P}\left( \displaystyle\lim_{k\to \infty} \frac{1}{N_k}\displaystyle\sum_{n=1}^{N_k} \mathbbm{1}_{\mu^{\eta}_{\Lambda_n} \in B^{\tau}_X(\mu^+)\cup B^{\tau}_X(\mu^-)}=1 \right)=1
$$
and it is enough to show that, for every $\delta'>0$,
$$
\mathbb{P}\left( \left| \displaystyle\frac{1}{N_k}\displaystyle\sum_{n=1}^{N_k} \mathbbm{1}_{\mu^{\eta}_{\Lambda_n} \in B^{\tau}_X(\mu^+)\cup B^{\tau}_X(\mu^-)}-1\right|>\delta' \ \text{i.o.} \right)=0.
$$

To apply Borel-Cantelli, we should prove
$$
\displaystyle\sum_{k=1}^{\infty}\mathbb{P}\left( \left| \displaystyle\frac{1}{N_k}\displaystyle\sum_{n=1}^{N_k} \mathbbm{1}_{\mu^{\eta}_{\Lambda_n} \in B^{\tau}_X(\mu^+)\cup B^{\tau}_X(\mu^-)}-1\right|>\delta' \right) <\infty.
$$

Note that, by Markov's inequality and by Proposition \ref{mainprop}, using $$\mathbb{P}(\mu^{\eta}_{\Lambda_n} \notin B^{\tau}_X(\mu^+)\cup B^{\tau}_X(\mu^-)) \le Dn^{-\alpha +\frac{1}{2}+\varepsilon}$$  for every $n\ge n_0$,
\begin{align*}
\mathbb{P}\left( \left| \displaystyle\frac{1}{N_k}\displaystyle\sum_{n=1}^{N_k} \mathbbm{1}_{\mu^{\eta}_{\Lambda_n} \in B^{\tau}_X(\mu^+)\cup B^{\tau}_X(\mu^-)}-1\right|>\delta' \right)
&= \mathbb{P}\left( \displaystyle\frac{1}{N_k}\displaystyle\sum_{n=1}^{N_k} \left(1-\mathbbm{1}_{\mu^{\eta}_{\Lambda_n} \in B^{\tau}_X(\mu^+)\cup B^{\tau}_X(\mu^-)}\right)>\delta' \right)\\
&= \mathbb{P}\left( \displaystyle\frac{1}{N_k}\displaystyle\sum_{n=1}^{N_k} \mathbbm{1}_{\mu^{\eta}_{\Lambda_n} \notin B^{\tau}_X(\mu^+)\cup B^{\tau}_X(\mu^-)}>\delta' \right)\\
&\le \frac{1}{\delta' N_k}\sum_{n=1}^{N_k}\mathbb{P}(\mu^{\eta}_{\Lambda_n} \notin B^{\tau}_X(\mu^+)\cup B^{\tau}_X(\mu^-))\\
&\le \frac{n_0}{\delta' N_k}+ \frac{D}{\delta'} N_{k}^{-\alpha+\frac{1}{2}-\varepsilon}.
\end{align*}
Using $N_k>k^{\left( \alpha-\frac{1}{2} \right)^{-1}+a}$ for a fixed $a>0$,
$$
\sum_{k=1}^{\infty}\left(\frac{n_0}{\delta' N_k}+ \frac{D}{\delta'} N_{k}^{-\alpha+\frac{1}{2}-\varepsilon} \right)<\infty.
$$
Therefore,
$$
\mathbb{P}\left( \displaystyle\lim_{k\to \infty} \frac{1}{N_k}\displaystyle\sum_{n=1}^{N_k} \mathbbm{1}_{\mu^{\eta}_{\Lambda_n} \in B}=1 \right)=1.
$$

By symmetry of the distribution, we have
$$ 
\mathbb{P}\left( \displaystyle\lim_{k\to \infty} \frac{1}{N_k}\displaystyle\sum_{n=1}^{N_k} \mathbbm{1}_{\mu^{\eta}_{\Lambda_n} \in B^{\tau}_X(\mu^+)}=\displaystyle\lim_{k\to \infty} \frac{1}{N_k}\displaystyle\sum_{n=1}^{N_k} \mathbbm{1}_{\mu^{\eta}_{\Lambda_n} \in  B^{\tau}_X(\mu^-)} \right)=1.
$$
Therefore, since the balls $B^{\tau}_X(\mu^+)$ and $B^{\tau}_X(\mu^-)$ are disjoint, 
$$
\mathbb{P}\left( \displaystyle\lim_{k\to \infty} \frac{1}{N_k}\displaystyle\sum_{n=1}^{N_k} \mathbbm{1}_{\mu^{\eta}_{\Lambda_n} \in B^{\tau}_X(\mu^+)}=\frac{1}{2} \right)=1
$$
and
$$
\mathbb{P}\left( \displaystyle\lim_{k\to \infty} \frac{1}{N_k}\displaystyle\sum_{n=1}^{N_k} \mathbbm{1}_{\mu^{\eta}_{\Lambda_n} \in B^{\tau}_X(\mu^-)}=\frac{1}{2} \right)=1.
$$
This implies that
$$
\mathbb{P}\left( \displaystyle\lim_{k\to \infty} \frac{1}{N_k}\displaystyle\sum_{n=1}^{N_k} \mathbbm{1}_{\mu^{\eta}_{\Lambda_n} \in B}=\frac{1}{2} \right)=1
$$
for every $B$ such that $\mu^{\pm}\in B$ and $\mu^{\mp}\notin B$.
\end{proof}

\section*{Appendix 2: Proof of Proposition \ref{toy_model_prop}}

\begin{proof}[Proof of Proposition \ref{toy_model_prop}]
Item (1): Take a subsequence $N_k=k^{2(\alpha-\frac{1}{2})^{-1}}$, for $\alpha > \frac{1}{2}$ so that $N_k$ is indeed a diverging increasing subsequence of integers.

Consider first  $\mu \not \in \{\delta_+,\delta_-\}$. Then there exists a weakly open set $U\subset \mathcal{M}_1(\Omega)$ such that $\mu\in U$ and $\delta_+,\delta_-\notin U$. Choosing a finite set $X\subset \Z$ and $\varepsilon>0$ such that 
$$
B_X^{\varepsilon}(\delta_+)\cap U = \emptyset \quad \text{and} \quad B_X^{\varepsilon}(\delta_-)\cap U = \emptyset.
$$
 We have
\beq 
\mathbb{P}\left(\mu^{\eta}_{\Lambda_{N_k}}\in U\right) &\le & \bbP\left( \mu^{\eta}_{\Lambda_{N_k}} \notin B_X^{\varepsilon}(\delta_+)\cup B_X^{\varepsilon}(\delta_-) \right) \nonumber\\
&= &\mathbb{P}\left( \lVert \mu^{\eta}_{\Lambda_{N_k}} - \delta_+ \rVert_X \land \lVert \mu^{\eta}_{\Lambda_{N_k}} - \delta_- \rVert_X \ge \varepsilon  \right)\nonumber\\
&\le& C \cdot N_k^{\frac{1}{2}-\alpha}=  C \cdot k^{-2}\label{Page15},
\eeq
where $C=C(\alpha)>0$ is a constant. For every $\theta\in (0,1)$, let us show that
\be\label{eq2}
\bbp\left( \kappa^{\emp}_{N_k}(U) >\theta \text{ i.o.} \right)=0.
\ee
We have
$$
\bbp\left( \kappa^{\emp}_{N_k}(U) >\theta \right) \le \frac{1}{\theta}\E( \kappa^{\emp}_{N_k}(U)) = \frac{1}{\theta N_k} \sum_{n=1}^{N_k} \E(\delta_{\mu_{\Lambda_n}}(U)).
$$
Since  $\E(\delta_{\mu_{\Lambda_n}}(U))= \bbp(\mu_{\Lambda_n} \in U)$,
$$
\bbp\left( \kappa^{\emp}_{N_k}(U) >\theta \right) \le\frac{1}{\theta N_k}\sum_{n=1}^{N_k}\frac{1}{n^{\alpha-\frac{1}{2}}} \le \frac{C}{\theta N_k}N_k^{\frac{3}{2}-\alpha} = \frac{C}{\theta}N_k^{\frac{1}{2}-\alpha}.
$$
Thus, for $N_k = k^{2(\alpha-\frac{1}{2})^{-1}}$,
$$
\sum_{k=1}^{\infty} \bbp\left( \kappa^{\emp}_{N_k}(U) >\theta \right) \le \frac{C}{\theta}\sum_{k=1}^{\infty}\frac{1}{k^2}<\infty.
$$
By the Borel-Cantelli Lemma, we conclude (\ref{eq2}). Therefore
$$
\bbp\left( \lim_{k\to\infty}\kappa^{\emp}_{N_k}(U) =0  \right)=1,
$$
i.e., $\kappa^{\emp}_{N_k}(U)$ converges to 0 almost surely.

Now, choose $X$ and $\eps$ such that $B^{\eps}_X(\delta_+) \cap B^{\eps}_X(\delta_-)=\emptyset$. We have to show
\be\label{eq3}
\bbp\left( \lim_{k\to\infty}\kappa^{\emp}_{N_k}(B^{\eps}_X(\delta_+)) =\frac{1}{2}  \right)=\bbp\left( \lim_{k\to\infty}\kappa^{\emp}_{N_k}(B^{\eps}_X(\delta_-)) =\frac{1}{2}  \right)=1.
\ee

Let us first show that 
 for every $\theta \in (0,1)$,
\be\label{eq3-1}
\bbp\left( \left| \kappa^{\emp}_{N_k}(B^{\eps}_X(\delta_+) \cup B^{\eps}_X(\delta_-)) -1\right| >
\theta  \text{ i.o.} \right)=0.
\ee
First, by the triangle inequality,
$$
\bbp\left( \left| \kappa^{\emp}_{N_k}(B^{\eps}_X(\delta_+) \cup B^{\eps}_X(\delta_-)) -1\right| >
\theta \right) \le \frac{1}{\theta N_k}\sum_{n=1}^{N_k}\E ( | \delta_{\mu_{\Lambda_n}}( B^{\eps}_X(\delta_+) \cup B^{\eps}_X(\delta_-) )-1 | ).
$$
Since
$$
\E ( | \delta_{\mu_{\Lambda_n}}( B^{\eps}_X(\delta_+) \cup B^{\eps}_X(\delta_-) )-1 | ) = 
\bbp (\mu_{\Lambda_n} \notin   B^{\eps}_X(\delta_+) \cup B^{\eps}_X(\delta_-)),
$$
and, that the latter is dominated by $C n^{1/2- \alpha}$) by (\ref{Page15}). So , choosing again $N_k = k^{2(\alpha-\frac{1}{2})^{-1}}$,
$$
\bbp\left( \left| \kappa^{\emp}_{N_k}(B^{\eps}_X(\delta_+) \cup B^{\eps}_X(\delta_-)) -1\right| >
\theta \right) \le 
\frac{1}{\theta N_k}\sum_{n=1}^{N_k}C n^{\frac{1}{2}-\alpha}\le \frac{C}{\theta N_k} N_k^{\frac{3}{2}-\alpha} =  \frac{C}{\theta} k^{-2}.
$$
and  (\ref{eq3-1}) holds for slow decays as we desired, by Borel-Cantelli lemma.
Thus, we get
$$
\bbp\left( \lim_{k\to\infty}\kappa^{\emp}_{N_k}(B^{\eps}_X(\delta_+) \cup B^{\eps}_X(\delta_-)) = 1  \right)=1.
$$
or equivalently
$$
\bbp\left( \lim_{k\to\infty}\kappa^{\emp}_{N_k}(B^{\eps}_X(\delta_+)) +  \kappa^{\emp}_{N_k} (B^{\eps}_X(\delta_-)) =1  \right)=1.
$$
Since the distribution of the random boundary condition is symmetric, we have
\be\label{eq4}
\kappa^{\emp}_{N_k}[-\eta](B^{\eps}_X(\delta_+)) = \kappa^{\emp}_{N_k}[\eta](B^{\eps}_X(\delta_-)).
\ee
so
\be\label{eq5}
\bbp\left( \lim_{k\to\infty}\kappa^{\emp}_{N_k}(B^{\eps}_X(\delta_+)) =\frac{1}{2}  \right)=\bbp\left( \lim_{k\to\infty}\kappa^{\emp}_{N_k}(B^{\eps}_X(\delta_-)) =\frac{1}{2}  \right).
\ee

All together, 
(\ref{eq3}) holds.

{\em 2. Intermediate decays $ 0 \leq \alpha < \frac{1}{2}$ : bounded boundary energies.}

In this regime there is actually no major difference between the arguments for the toy model, which provide a metastate supported by mixed ground states, and those for the positive-low-temperature model (described in next Section, item (b) of Theorem \ref{thm1}, which provides  a "new" type of metastate supported by mixed Gibbs measures.

In both cases, the almost surely finiteness, asymptotically,  of the  boundary energies along  subsequences provides at low temperatures absolute continuity of the limiting measures in the thermodynamc limit. The proof uses an adaptation of the argument of "Equivalence of boundary conditions" of  \cite{BLP}.
\end{proof}

\section*{Appendix 3: Proof of Theorem \ref{thm1}}

\begin{itemize}
\item Item (a) is a direct consequence of Theorem \ref{thma}, proven in Section 7.

\item Let us prove item (b). Although it is the main new phenomenon occurring in our long-range model with random boundary condition, its proof appears to be less painful than proving item (a), which is an extension to long-range of the higher dimensional ($d \geq 2$) behaviors in the n.n. case. Here, for intermediate decays, a.s. convergence of the random boundary energies allows to proceed without invoking the contour machinery, by more abstract equivalence of b.c. type arguments.

We will say that two boundary conditions left of the origin,  $\eta$ and $\eta'$, have a \emph{finite energy difference} if
\be \label{FiniteNRJ}
C_{\eta,\eta'}=\sup_{N \in \mathbb{N}}\sup_{\sigma \in \Omega_{\Lambda_N}}\left| W^{\eta}_N(\sigma) - W^{\eta'}_N(\sigma) \right|<\infty.
\ee

For any Gibbs measure $\rho \in \mathcal{G}(\gamma)$, there exists a unique probability measure $\mu_{\rho}$ concentrated on $\ex \mathcal{G}(\gamma)$, such that
$$
\rho= \int_{\ex\mathcal{G}(\gamma)} \nu\mu_{\rho}(d\nu).
$$
Two Gibbs states $\rho$ and $\rho$ are \emph{equivalent} (mutually absolutely continuous) if and only if $\mu_{\rho}$ and $\mu_{\rho'}$ are equivalent. This means that they have the same support. 

The following result \cite{BLP} proves that two Gibbs measures with  b.c. $\eta$ and $\eta'$ of finite energy difference have the same support in their mixtures.

\begin{theorem}
Given two
boundary conditions
 $\eta$ and $\eta'$ having a finite energy difference and a sequence $\Lambda_n\uparrow \mathbb{Z}$ such that
$$
\lim_{n\to \infty}\mu^{\eta}_{\Lambda_n}=\rho, \quad\quad \lim_{n\to \infty}\mu^{\eta'}_{\Lambda_n}=\rho',
$$
then $\mu$ and $\mu'$ are equivalent and so are $\mu_{\rho}$ and $\mu_{\rho'}$. In particular, if $\rho$ is extremal, $\rho=\rho'$.
\end{theorem}


In our simple situation all Gibbs measures are of the form 
$$
\mu = \lambda \mu^{+} +(1-\lambda) \mu{-}
$$
When one has free boundary conditions, the limit  Gibbs measure is symmetric and thus one has $\lambda= \frac{1}{2}$.
When one has that $W^\eta$ is finite, almost surely, this means that the corresponding weights $\lambda[\eta]$ are different from $0$ and $1$, almost surely. Thus the limit states on which the metastate lives are the mixed ones. The weights with which the different mixtures occur can depend on $\alpha$ and the temperature $\beta$, see the decomposition (\ref{decompo}) or more formally
\begin{equation}\label{expwe}
\mu^{\eta}_\Lambda = \lambda_\Lambda[\eta] \cdot \mu_\Lambda^+ + \big(1-\lambda_\Lambda[\eta]\big) \cdot  \mu_\Lambda^-
\end{equation} and the formal weights 
\be \label{weights}
\lambda_\Lambda[\eta]\sim \frac{e^{-\beta F(W_{\Lambda}^\eta  )}}{e^{-\beta F(W_{\Lambda}^\eta)} + e^{\beta F(W_{\Lambda}^\eta)}}=\frac{1}{1+e^{-2\beta F(W_{\Lambda}^\eta)}}.
\ee

This is the case for the intermediate range of decays  $0 \leq \alpha < \frac{1}{2}$. For such $\alpha$'s, consider typical random b.c. $\eta$ and $\eta'$ with finite random boundary energy differences, then   (\ref{FiniteNRJ}) holds.:
\begin{equation}\label{FiniteNRJDiff}
C_{\eta,\eta'}= \sup_{N \in \mathbb{N}}\sup_{\sigma \in \Omega_{\Lambda_N}} \Big| W_N^\eta(\sigma) - W_N^\eta(\sigma) \Big| < \infty
\end{equation}

To see that this case leads to equivalent infinite-volume limits of $\mu_\Lambda^\eta$ and $\mu_{\Lambda}^{\eta'}$ -- and thus our  result of a fully dispersed metastate on mixed states in this regime --, consider another volume $\Lambda'$ finite and $E_{\Lambda'}$ a cylinder on it. Then the ratios of partition functions at different volumes $\Delta=\Lambda,\Lambda'$ satisfy the bound
$$
\frac{Z_\Delta^\eta}{Z_\Delta^{\eta'}} = \frac{\int_{\Omega_\Delta} e^{-\beta \big (H_\Delta^\eta(\sigma) - H_\Delta^{\eta'}(\sigma) \big)}  e^{-\beta H_\Delta^{\eta'} (\sigma)} \d \sigma_\Delta}{\int_{\Omega_\Delta}  e^{-\beta H^{\eta'}_\Delta (\sigma) \d \sigma_\Delta} }\leq e^{\beta C_{\eta,\eta'}}
$$
to eventually get, as in \cite{BLP}, that for any $\Lambda \supset \Lambda'$
$$
e^{-2 \beta C_{\eta,\eta'} }\mu_\Lambda^\eta(E_{\Lambda'}) \leq \mu_{\Lambda}^{\eta'}(E_{\Lambda'}) \leq e^{-2 \beta C_{\eta,\eta'}} \mu_\Lambda^\eta(E_{\Lambda'}) 
$$
proving our claims. 

\section*{Appendix 4:  Mayer expansion of the characteristic function}


Let us first sketch how one needs and uses a low-temperature Mayer expansion to conclude by using a WLLT, as described in various parts  of the paper.

For free boundary conditions there is a convergent low-temperature expansion. Contours have an energy growing with their length, with a constant about half of that for pure boundary conditions.  So the low-temperature expansion converges for temperatures low enough. (This  works immdiately  for $\alpha > \alpha^{*}$, to go  beyond that the extensions due to Kimura or Littin-Picco, and we also provide the extension in Section 6). The sum over all configurations is twice the sum over all contour configurations  and thus one writes $ Z_N^{f} = Z_{N,+}^{f} + Z_{N,-}^{f} = 2 Z_{N,+}^{f}$.
The subscripts $+$ and $-$ indicate the signs of the spins external to the contours. Boundary condition $\eta$ on volume $\Lambda$ induces equivalently a collection of correlated external fields $h_i(\eta) = \sum_{j \in \Lambda^c} J(i-j) \eta_j $. If $i$ is the distance to the boundary,  $|h_i|= O(|i|^{\alpha - \frac{3}{2}}) $.

Let a boundary condition be "$k$-good" if all the local fields at distance $j$,  with $ j > n_k$,  from the boundary are in absolute strength less than $| j|^{\alpha - \frac{3}{2} + \varepsilon} $. We will choose $n= n(k)$ large enough, so that the probability of NOT being $k$-good is less than $\frac{1}{k^2}$. 
This is possible due to  an exponential Chebyshev inequality and a union bound, and  summing over $k$ afterwards allows us the use of Borel-Cantelli arguments. Moreover we require that adding the truncated field terms, that is all terms 
 at distance larger than $n_k$ from the boundary for both $Z_{N,+}^{\eta}$ and $Z_{N,-}^{\eta}$ results in a convergent cluster expansion by the arguments of \cite{BEEKR}. This corresponds to  the Mayer expansion part of \cite{ENS1}. Notice that the decay suffices --the Imry-Ma regime of \cite{BEEKR} is covered-- and that due to the truncation a weakness condition holds, while we choose later the volume $N=N_k$, dependent on $n_k$ to be large enough. 
By Borel-Cantelli there will be a measure one subset of good boundary conditions, each of which will be $k$-good for a certain $k$ and have that $n^2 < N^{\delta}$.\\

 Thus we will have that the truncated difference in free energies $F_N^{\eta}$  between $+$-phase  and $-$-phase, which is a difference of terms of the form 
$\mu_{N, \pm}^{f} \big[e^{W_{n,N}^{\eta} }\big]$, will be outside a window of size
 $N^{\epsilon}$  with probability $N^{\alpha - \frac{1}{2} + \epsilon}$. So the probability that  $\eta$ is good is close to 1, and for a good $\eta$ the average over such $\eta$ of 
$e^{i s  F^\eta_N}$  is of leading order $e^{- s^2 N^{2 \alpha -1}}$, which gives the WLLT.  We can integrate over an $s$-interval which does not extend to infinity,  but is large enough to give the Gaussian answer. Then we get a Mayer expansion to implement the WLLT as follows.

Let us precise a bit more with the notations of Section 6 (cluster expansion with decoupling intervals). Since
$$
\log \tilde{\Xi}^{\eta}_{N_k,n_k} =\sum_{n=1}^{\infty} \frac{1}{n!}\sum_{R_1}\ldots\sum_{R_n}\phi^T(R_1,\ldots,R_n)\prod_{i=1}^n\xi^{\eta}_{N_k,n_k}(R_i)
$$
is absolutely convergent for sufficiently large $\beta$ and restricted to good boundary conditions $
\eta\in \Omega_0$, we have
$$
\log \tilde{\Xi}^{\eta}_{N_k,n_k}  = \sum_{R}\xi^{\eta}_{N_k,n_k}(R)(1+\Pi^{\eta}_{N_k,n_k}(R)),
$$
where
$$
\Pi^{\eta}_{N_k,n_k}(R) = \sum_{n=2}^{\infty}\frac{1}{n!}\sum{R_2}\ldots\sum_{R_n}\phi^T(R,R_2,\ldots,R_n) \prod_{i=2}^n\xi^{\eta}_{N_k,n_k}(R_i).
$$
Thus,
\begin{align*}
F^{\eta}_{N_k,n_k}&=\log \tilde{\Xi}^{\eta}_{N_k,n_k} -\log \tilde{\Xi}^{\eta}_{N_k,n_k}\\
&=\sum_{R}(\xi^{\eta}_{N_k,n_k}(R)(1+\Pi^{\eta}_{N_k,n_k}(R))-\xi^{-\eta}_{N_k,n_k}(R)(1+\Pi^{-\eta}_{N_k,n_k}(R))).
\end{align*}
Define
$$
U^{\eta}_{N_k,n_k}(R)=\xi^{\eta}_{N_k,n_k}(R)(1+\Pi^{\eta}_{N_k,n_k}(R))-\xi^{-\eta}_{N_k,n_k}(R)(1+\Pi^{-\eta}_{N_k,n_k}(R)).
$$
We have
\begin{align*}
\psi(t)= \E\lf e^{it W^{\eta}_{N_k}}e^{itF^{\eta}_{N_k,n_k}} e^{it\bar{F}^{\eta}_{N_k,n_k}} \ri.
\end{align*}
Let us control the middle term, applying the Mayer expansion,
\begin{align*}
\psi_{N_k,n_k}(t)&= \E\lf e^{itW^{\eta}_{N_k}+itF^{\eta}_{N_k,n_k}} \ri\\
&=\E\lf e^{itW^{\eta}_{N_k}}\prod_{R}e^{itU^{\eta}_{N_k,n_k}(R)} \ri\\
&=\E\lf e^{itW^{\eta}_{N_k}}\sum_{\calr} \prod_{R\in \calr}(e^{itU^{\eta}_{N_k,n_k}(R)} -1) \ri\\
&=\sum_{\calr}\E\lf  e^{itW^{\eta}_{N_k}} \prod_{R\in \calr}(e^{itU^{\eta}_{N_k,n_k}(R)} -1)\ri\\
&=\sum_{\calr}\E\lf  e^{itW^{\eta}_{N_k}} \prod_{R\in \calr}(e^{itU^{\eta}_{N_k,n_k}(R)} -1) \mathbbm{1}_{\Omega_0}\ri.
\end{align*}
Define
$$
w_{N_k,n_k}(\calr)=\E\lf  e^{itW^{\eta}_{N_k}} \prod_{R\in \calr}(e^{itU^{\eta}_{N_k,n_k}(R)} -1) \mathbbm{1}_{\Omega_0}\ri.
$$
Since the model is long-range, the weight $w$ is not multiplicative as in the nearest neighbor model, but we have derived equivalent cluster expansions with appropriate decoupling intervals in Section 6. 
\end{itemize} 


We have
$$
\mu^{\eta}_{\Lambda_N,\beta}(\sigma) = \mu^f_{\Lambda_N,\beta}\left( e^{-\beta W^{\eta}_{\Lambda_N}(\sigma)} \right).
$$
Define, for $C\subset \Lambda_N$ and $y\notin \Lambda_N$,
$$
h_C(y) = \sum_{x\in C}J_{xy}\sigma_x.
$$
Also, for a set $A\subset \Lambda^c_N$, define
$$
Z^{\eta}_{N,\beta}(C|A) = \sum_{\sigma\in \Omega_{\Lambda_N}} e^{-\beta H_{\Lambda_N}(\sigma)} \prod_{y\in A}e^{\beta h_C(y)\eta_y}
$$ 
and
$$
Z^{\eta}_{N,\beta}(C|\emptyset) = Z_{\Lambda_N,\beta},
$$
that is, the partition function with free boundary condition. Define, for $n>N$
$$
B_n = \{-n,-n+1,\ldots,-N-1\} \cup \{N+1,\ldots, n-1,n\},
$$
and $B_N=\emptyset$.
Note that, by the telescopic sum,
$$
\mu^{\eta}_{\Lambda_N,\beta}(\sigma)  = \prod_{n=N}^{\infty} \frac{Z^{\eta}_{N,\beta}(\Lambda_N|B_{n+1})}{Z^{\eta}_{N,\beta}(\Lambda_N|B_{n})}
$$
The numerator and denominator are partition functions $Z_{N,\beta}(\Lambda_N|B_{n})$ with boundary condition $\eta$ on $B_n$.

We will use the method of telescopic sums to decompose the partition function 
$$
\tilde{\Xi}^{\eta}_{N_k,n_k} = \sum_{\sigma\in \Omega^+_{\Lambda_{N_k}}} \exp\left( -\beta H_{\Lambda_{N_k}}(\sigma) -2\beta \sum_{x\in \Lambda_{N_k}}\tilde{h}^{\eta}_{N_k,n_k}(x)\mathbbm{1}_{\sigma_x=-1}\right).
$$
Since
$$
\tilde{h}^{\eta}_{N_k,n_k}(x) = \sum_{y\in \Lambda_{N_k}^c}J_{xy}\eta_y \quad \text{ for }x\in [-N_k+n_k,N_k-n_k]
$$
and 0 otherwise, let $C_{N_k,n_k} = [-N_k+n_k,N_k-n_k]$, we have
$$
\tilde{\Xi}^{\eta}_{N_k,n_k} = \tilde{\Xi}^{\eta}_{N_k,\beta}(C_{N_k,n_k}| \Lambda^c_N) = \prod_{n=N_k+1}^{\infty} \frac{\tilde{\Xi}^{\eta}_{N_k,\beta}(C_{N_k,n_k}| B_{n})}{\tilde{\Xi}^{\eta}_{N_k,\beta}(C_{N_k,n_k}| B_{n-1})}.
$$
It is also important to see that
$$
\frac{\tilde{\Xi}^{\eta}_{N_k,\beta}(C_{N_k,n_k}| B_{n})}{\tilde{\Xi}^{\eta}_{N_k,\beta}(C_{N_k,n_k}| B_{n-1})} = \tilde{\mu}^{\eta_{B_{n-1}}}_{N_k,n_k} \left( \prod_{y\in \{-n,n\}}e^{\beta \eta_y\sum_{x\in C_{N_k,n_k}}J_{xy}\sigma_x} \right).
$$

Thus, the characteristic function can be written as
\begin{align*}
\psi_{N_k}(t) = \mathbb{E}\left( \prod_{n={N_k+1}}^{\infty}e^{it \sum_{y\in \{-n,n\}}W_{N_k}^{\eta_y}} e^{\frac{it}{2\beta}\mathfrak{T}^{\eta}_{N_k,n_k}(B_{n-1})} e^{\frac{it}{2\beta}\mathfrak{S}^{\eta}_{N_k,n_k}}\right),  
\end{align*}
where
\small
$$
\mathfrak{T}^{\eta}_{N_k,n_k}(B_{n-1})=\log \tilde{\mu}^{\eta_{B_{n-1}}}_{N_k,n_k} \left( \prod_{y\in \{-n,n\}}e^{\beta \eta_y\sum_{x\in C_{N_k,n_k}}J_{xy}\sigma_x} \right)-\log \tilde{\mu}^{-\eta_{B_{n-1}}}_{N_k,n_k} \left( \prod_{y\in \{-n,n\}}e^{\beta \eta_y\sum_{x\in C_{N_k,n_k}}J_{xy}\sigma_x} \right), 
$$
\normalsize
and
$$
\mathfrak{S}^{\eta}_{N_k,n_k} = \log \frac{\tilde{\Xi}^{\eta}_{N_k}}{\tilde{\Xi}^{\eta}_{N_k,n_k}} -\log \frac{\tilde{\Xi}^{-\eta}_{N_k}}{\tilde{\Xi}^{-\eta}_{N_k,n_k}}. 
$$

{\bf Remark:} Note that, for the toy model, we have
$$
\psi_{N_k}(t)  =\mathbb{E}\left( \prod_{n={N_k+1}}^{\infty}e^{it \sum_{y\in \{-n,n\}}W_{N_k}^{\eta_y}}\right) = \prod_{n={N_k+1}}^{\infty}\mathbb{E}\left( e^{it \sum_{y\in \{-n,n\}} \eta_y\sum_{x\in \Lambda_k}J_{xy}}\right).
$$
Note that
$$
\sum_{x\in \Lambda_k}J_{xn} = \sum_{x\in \Lambda_k}J_{x(-n)}
$$
and
$$
\mathbb{E}\left( e^{it \sum_{y\in \{-n,n\}} \eta_y\sum_{x\in \Lambda_k}J_{xy}}\right) = \frac{1}{2} + \frac{1}{2}\cos\left(2t \sum_{x\in \Lambda_k}J_{xy}\right) = \cos^2\left(t \sum_{x\in \Lambda_k}J_{xy}\right).
$$
Thus, the proof of the toy model still holds with correcting terms, since the square does not affect the conclusion.

{\bf Acknowledgments:} We thank Pierre Picco for stimulating discussions and his encouragement  to study Dyson models with random boundary conditions during his visit  to EURANDOM (TU/e, Eindhoven) in May 2018. A.v.E. thanks Roberto Fern\'andez and NYU Shanghai for a very inspiring visit when we made our first main progress. Further steps occurred when  A.v.E. and A.L.N. met in Cr\'eteil, Brest,  Leiden, S\~{a}o Carlos and Oberwolfach. A.v.E. also thanks Evgeny Verbitskiy and Anders Oberg for stimulating exchanges, and A.v.E. and A.L.N.  thank the participants of the Oberwolfach  Mini-Workshop: One-sided and Two-sided Stochastic Descriptions for useful discussions. E.E. thanks Vlad Margarint for fruitful discussions. and eventually A.L.N and E.E. thank NYU Shanghai for the support of a two-weeks visit crucial to conclude this work, while an intermediate Research in Pairs stay in Paris, via an IHP RIP grant in June 2023, has also made it possible.

 The research of A.L.N. has also been partly supported by {\em Labex B\'ezout} at Universit\'e Paris Est (UPE), funded by ANR (reference ANR-10-LABX-58) and by the CNRS   {\em International Research Project B\'ezout-EURANDOM}.

\end{document}